\documentclass[a4paper,11pt,reqno]{amsart}

\usepackage[english]{babel}
\usepackage[utf8]{inputenc}		
\usepackage[scaled]{helvet}		
\usepackage{courier}			
\usepackage{eulervm}			
\usepackage[bb=boondox,bbscaled=1.05]{mathalfa}	%
\usepackage{dsfont}
\normalfont

\usepackage{pdfsync}
\usepackage{tabularx,booktabs}			
\usepackage{caption,subcaption}			
\usepackage[dvipsnames]{xcolor}			
\usepackage{bookmark}					
\usepackage{textcomp}
\usepackage{enumerate}
\usepackage{leftidx,tensor,accents}
\usepackage[margin=1in]{geometry}

\usepackage{titlesec}
\titleformat{\section}
	{\centering\normalfont\scshape\bfseries}{\thesection.}{1em}{} 
\titleformat{\subsection}
 	{\centering\normalfont\bfseries}{\thesubsection.}{1em}{} 
\titleformat*{\paragraph}{\bfseries}

\makeatletter
	\def\l@section{\@tocline{1}{0pt}{0pc}{1pc}{}} 
	\def\l@subsection{\@tocline{2}{0pt}{1pc}{2pc}{}} 
\makeatother

\makeatletter
\def\@tocline#1#2#3#4#5#6#7{\relax
	\ifnum #1>\c@tocdepth 
	\else
	\par \addpenalty\@secpenalty\addvspace{#2}%
	\begingroup \hyphenpenalty\@M
	\@ifempty{#4}{%
		\@tempdima\csname r@tocindent\number#1\endcsname\relax
	}{%
		\@tempdima#4\relax
	}%
	\parindent\z@ \leftskip#3\relax \advance\leftskip\@tempdima\relax
	\rightskip\@pnumwidth plus4em \parfillskip-\@pnumwidth
	#5\leavevmode\hskip-\@tempdima
	\ifcase #1
	\or\or \hskip 1em \or \hskip 2em \else \hskip 3em \fi%
	#6\nobreak\relax
	\dotfill\hbox to\@pnumwidth{\@tocpagenum{#7}}\par
	\nobreak
	\endgroup
	\fi}
\makeatother

\usepackage{graphicx}		
\usepackage{float}
\usepackage{tikz,tikz-cd}			
\usetikzlibrary{decorations.pathreplacing,calligraphy}
\usepackage{pgfplots}
\pgfplotsset{compat=1.18}

\usepackage[textsize=footnotesize,textwidth=0.85in]{todonotes}
\presetkeys%
	{todonotes}%
	{color=Dandelion}{}%

\usepackage{amsmath,amssymb,amsthm,mathtools}
\usepackage{bm,braket,faktor,stmaryrd}
\SetSymbolFont{stmry}{bold}{U}{stmry}{m}{n}
\numberwithin{equation}{section}

\usepackage[
		style = alphabetic,%
		maxalphanames = 5,%
		giveninits = true,%
		sorting = nyt,%
		maxbibnames = 99,%
		backend = bibtex%
	]{biblatex}
\renewbibmacro{in:}{} 
\usepackage{csquotes}
\usepackage{hyperref}					
\definecolor{webbrown}{rgb}{0.65, 0.16, 0.16}
\hypersetup{
colorlinks=true, linktocpage=true, pdfstartpage=1, pdfstartview=FitV,breaklinks=true, pdfpagemode=UseNone, pageanchor=true, pdfpagemode=UseOutlines,plainpages=false, bookmarksnumbered, bookmarksopen=true, bookmarksopenlevel=1,hypertexnames=true, pdfhighlight=/O,urlcolor=webbrown, linkcolor=RoyalBlue, citecolor=ForestGreen}

\usepackage[noabbrev,capitalize]{cleveref}
\expandafter\def\csname ver@etex.sty\endcsname{3000/12/31}

\Crefname{lemma}{Lemma}{Lemmata}
\Crefname{subsection}{Subsection}{Subsections}
\Crefname{conjecture}{Conjecture}{Conjectures}
\Crefname{question}{Question}{Questions}
\Crefname{warning}{Warning}{Warnings}

\newcommand{\breakcell}[2][c]{%
	\begin{tabular}[#1]{@{}c@{}}#2\end{tabular}}

\newcommand{\mf}[1]{\mathfrak{#1}}
\newcommand{\mc}[1]{\mathcal{#1}}
\newcommand{\mr}[1]{\mathrm{#1}}
\newcommand{\mb}[1]{\mathbb{#1}}
\newcommand{\ms}[1]{\mathsf{#1}}

\newcommand{\bbraket}[1]{\llbracket #1 \rrbracket}
\newcommand{\id}{\mathord{\mr{id}}}

\newcommand{\ord}{\mathord{\mr{ord}}}
\newcommand{\del}{\partial}
\newcommand{\dd}{\mr{d}}

\DeclareMathOperator*{\Res}{Res}

\newcommand{\Z}{\mb{Z}}

\newcommand{\R}{\mb{R}}
\newcommand{\C}{\mb{C}}
\renewcommand{\P}{\mb{P}}

\newcommand{\Mbar}{\overline{\mc{M}}}
\newcommand{\End}{\mathord{\mr{End}}}
\newcommand{\Spec}{\mathord{\mr{Spec}}}

\newcommand{\PIL}[1]{{}^{\textup{PI}}{#1}}
\newcommand{\PIR}[1]{{#1}^{\textup{PI}}}
\newcommand{\AiL}[1]{{}^{\textup{Ai}}{#1}}
\newcommand{\AiR}[1]{{#1}^{\textup{Ai}}}

\theoremstyle{plain}
\newtheorem{theorem}{Theorem}[section]

\newtheorem{proposition}[theorem]{Proposition}

\newtheorem{lemma}[theorem]{Lemma}
\newtheorem{corollary}[theorem]{Corollary}
\newtheorem{remark}[theorem]{Remark}
\newcommand{\thistheoremname}{}
\newtheorem*{genericthm*}{\thistheoremname}
\newenvironment{namedthm*}[1]
	{\renewcommand{\thistheoremname}{#1}%
		\begin{genericthm*}}
	{\end{genericthm*}}

\theoremstyle{definition}
\newtheorem{definition}[theorem]{Definition}

\title{The factorial growth of topological recursion}
\author[G.~Borot]{Ga{\"e}tan Borot}%
\address[G.~Borot]{
	Humboldt-Universit{\"a}t zu Berlin, Institut f{\"u}r Mathematik und Institut f{\"u}r Physik, Unter den Linden 6, 10099 Berlin, Germany}%
\email{gaetan.borot@hu-berlin.de}

\author[B.~Eynard]{Bertrand Eynard}%
\address[B.~Eynard]{
	Universit{\'e} Paris-Saclay, CNRS, CEA, Institut de Physique Th{\'e}orique, Gif-sur-Yvette, France %
	\& %
	Centre de Recherches Math{\'e}matiques de Montr{\'e}al, Universit{\'e} de Montr{\'e}al, QC, Canada
}%
\email{bertrand.eynard@ipht.fr}

\author[A.~Giacchetto]{Alessandro Giacchetto}%
\address[A.~Giacchetto]{
	ETH Z{\"u}rich, Departement Mathematik, Z{\"u}rich, Switzerland
}%
\email{alessandro.giacchetto@math.ethz.ch}

\begin{document}
\renewcommand{\hbar}{\hslash}

\begin{abstract}
	We show that the $n$-point, genus-$g$ correlation functions of topological recursion on any regular spectral curve with simple ramifications grow at most like $(2g - 2 + n)!$ as $g \rightarrow \infty$, which is the expected growth rate. This provides, in particular, an upper bound for many curve counting problems in large genus and serves as a preliminary step for a resurgence analysis.
\end{abstract}

\maketitle
\tableofcontents

\thispagestyle{empty}

\section{Introduction}
\label{sec:intro}

\subsection{Incipit}
The enumerative theory of Riemann surfaces spans a rich landscape of disciplines, including algebraic geometry, combinatorics, and theoretical physics. Problems such as computing volumes of moduli spaces of curves and intersection numbers, enumerating covers and triangulations of Riemann surfaces, and computing correlation functions in 2D quantum gravity and topological string theory, among others, share many interrelations and common structures. In particular, their complexity escalates rapidly with the genus. In quantum field theory, the growth is typically factorial in the number of loops, whereas in string theory, the contribution of worldsheets of genus $g$ is expected to be of order $(2g)!$ \cite{She91}. In the same complexity class as the latter, one encounters the problem of computing topological expansions of correlation functions in random matrix theory, various supersymmetric gauge theories, minimal models, and more generally, the semiclassical expansion of solutions to Lax-integrable PDEs. This was intensively discussed in theoretical physics in the early 1990s; see, for example, \cite{DS90,GM90,BK90,GZJ90} and the review \cite{DGZ95}.

In mathematics, the computation of Weil--Petersson volumes of the moduli space of hyperbolic surfaces of genus $g$ serves as a model case. In the early 2000s, Grushevsky and Schumacher--Trapani \cite{Gru01,ST01} established an upper bound of $(2g)!$ for the large genus Weil--Petersson volumes. Subsequent advancements by Mirzakhani and Zograf \cite{Zog,MZ15} improved the large genus bound to an asymptotic equivalent, up to an overall constant, which remains conjectural. The Mirzakhani--Zograf asymptotics have gained significant attention due to their profound implications in random hyperbolic geometry, where some quantities are only accessible in the large genus limit. Examples include the average value of the systole function \cite{MP19} and the (still conjectural) typical near-optimal spectral gap \cite{Wri20}. Similar asymptotic results have been obtained beyond Weil--Petersson volumes, such as for Masur--Veech volumes and Siegel--Veech constants \cite{Agg20,CMSZ20}, intersection numbers of $\psi$-classes on the moduli space of curves \cite{Agg21,DGZZ21,EGGGL}, and the number of cubic maps \cite{GK20}.

In fact, most of these problems are governed by the topological recursion formalism, originating in the work of Chekhov, Eynard, and Orantin. This formalism provides a recursion based on $2g - 2 + n$, which is equivalent to a set of Virasoro constraints for the quantities of interest (where $n$ is the number of boundaries or marked points). The initial data and coefficients of the recursion are determined by the spectral curve. The Virasoro constraints for the Airy spectral curve were the starting point for the large genus asymptotics of $\psi$-class intersection numbers derived by Aggarwal \cite{Agg21}. A natural question is how the large genus asymptotics of the topological recursion associated with a general spectral curve can be determined by a uniform method.

A possible way to approach the problem is through resurgence theory. Resurgence addresses the large genus asymptotic problem by analysing the singularities of the Borel transform of the all-genera generating series and aims to produce not only an asymptotic equivalent but also asymptotic expansions, including exponentially small corrections (transseries). We refer to \cite{MS16,ABS19,Mar} for recent introductions to resurgence and its applications. In theoretical physics, numerous problems of this kind have been treated (see, e.g., \cite{MSW08,ASV12,GIKM12,CESV16,BSSV23,GM23,GKKM24}), offering fine predictions supported by numerical studies of large genus behaviour, which often remain challenging to justify mathematically. The emerging picture suggests that all the necessary information concerning these singularities can be deciphered from the geometry of the spectral curve, following the strategy of \cite{EGGLS24,EGGGL}. Yet, a prerequisite for the resurgence analysis is determining the factorial growth rate (also called the Gevrey index) of the genus $g$ coefficient, or at least establishing an upper bound for it. This allows for the appropriate definition of the Borel transform and establishes its analyticity in a small disc around the origin in the Borel plane. While topological recursion is generally expected to exhibit a factorial growth of $(2g)!$, only a sub-optimal $(5g)!$ bound has been established so far \cite{Eyn}. The present work justifies the $(2g)!$ expectation in the form of an upper bound and provides stronger upper bounds that are uniform in $g$ and $n$.

\subsection{Main results}
Consider a spectral curve $\mc{S} = (\Sigma, x, y, \omega_{0,2})$, where $\Sigma$ is a smooth complex curve, $x$ has finitely many simple ramification points $a \in \mf{a}$ at which $\dd y$ is holomorphic, and $\omega_{0,2}$ is a symmetric meromorphic bi-differential having a double pole on the diagonal and no residue. The output of topological recursion on $\mc{S}$ is a collection of meromorphic symmetric multidifferentials $(\omega_{g,n})_{(g,n)}$ on $\Sigma^n$, called \emph{correlators}, having poles at ramification points and a collection of numbers $(\omega_{g,0})_{g \geq 2}$, called \emph{free energies}. All definitions will be reviewed in \cref{subsec:TR:SC}.

\begin{namedthm*}{Main theorem} \label{thm:main}
	Let $Z$ be local coordinate in an open set $U$ of $\Sigma \setminus \mf{a}$. For any compact $K \subset U$, there exist constants $\ms{T},\ms{B},\ms{C} > 0$ depending on the spectral curve, the local coordinate, and the compact set, such that the correlators obey an uniform upper bound for any $g \geq 0$ and $n > 0$ such that $2g - 2 + n > 0$:
	\begin{equation}
		\forall z_1,\ldots,z_n \in K
		\qquad
		\bigg|\frac{\omega_{g,n}(z_1,\ldots,z_n)}{\dd Z(z_1) \cdots \dd Z_n(z_n)}\bigg|
		\leq
		\ms{T} \, \frac{(3g - 3 + 2n)!}{\mathsf{B}^g \, \ms{C}^n \, g! \, n!} \,.
	\end{equation}
	If we fix $n > 0$, there exist constants $\ms{S}_{n},\ms{A} > 0$ such that
	\begin{equation}
		\bigg|\frac{\omega_{g,n}(z_1,\ldots,z_n)}{\dd Z(z_1) \cdots \dd Z_n(z_n)}\bigg|
		\leq
		\ms{S}_{n} \frac{\Gamma(2g - 2 + n)}{\ms{A}^{2g - 2 + n}} \,.
	\end{equation}
	Likewise, for the free energies, there exist constants $\ms{S}_0,\ms{A}_{0}$ depending only on the spectral curve, such that for any $g \geq 2$
	\begin{equation}
		|\omega_{g,0}| \leq \ms{S}_0 \frac{\Gamma(2g - 2)}{\ms{A}_0^{2g - 2}} \,.
	\end{equation}
\end{namedthm*}

This result appears in the text as \cref{cor:cor} and \cref{thm:free}. In applications of topological recursion to enumerative problems, the quantities of interest can be extracted by computing \emph{generalised periods} of the correlators, i.e. $(I_1 \otimes \cdots \otimes I_n)[\omega_{g,n}]$ for specific linear forms $I_i$ on the space of meromorphic differentials on $\Sigma$ with poles at the ramification points. Under weak assumptions on these linear forms, the generalised periods satisfy similar bounds (cf. \cref{thm:bound:periods}). A remarkable consequence is that the generating series in $\hbar$ 
\begin{equation}
	\sum_{g \geq 0} \hbar^{2g - 2 + n} (I_1 \otimes \cdots \otimes I_n)[\omega_{g,n}]
	\qquad\quad\text{or}\qquad\quad
	\sum_{\substack{g \geq 0 \\ n > 0}} \frac{\hbar^{2g - 2 + n}}{n!} I^{\otimes n}[\omega_{g,n}]
\end{equation}
are Gevrey-$1$. In particular the wave function, obtained when $I$ is the integration along a path in $\Sigma \setminus \mf{a}$, is Gevrey-$1$ (\cref{thm:wave}).

The values we obtain for the exponential growth rate of the upper bounds (the constants $\ms{A}, \ms{B}, \ms{C}$) can be computed from the geometry of the spectral curve, but they are not optimal. As an illustration, a non-exhaustive list of applications is presented in \cref{tbl:exmpls}, and we describe the corresponding upper bounds with explicit exponential growth rates in \cref{sec:exmpls}.

\begin{table}
\centering
\begin{tabular}[t]{cc}
	\toprule
	Spectral curve &
	Enumerative problem \\
	\bottomrule
		$x(z) = z^2, \quad y(z) = -z/2$ &
		\breakcell{$\psi$-class intersection numbers, \\ metric ribbon graphs} \\
	\midrule 
		$x(z) = z^2, \quad y(z) = -\sin(2\pi z)/4\pi$ &
		Weil--Petersson volumes \\
	\midrule
		$x(z) = z + z^{-1}, \quad y(z) = -z$ &
		\breakcell{Euler characteristic of $\mc{M}_{g,n}$, \\ integral metric ribbon graphs} \\
	\midrule
		\breakcell{
			$x(z) = z^2, \quad y(z) = - z/2$ \\[1ex]
			$\omega_{0,2}(z_1,z_2) = \Bigl( \frac{1}{(z_1 - z_2)^2} + \frac{\pi^2}{\sin^2(\pi(z_1 - z_2))} \Bigr) \frac{\dd z_1 \dd z_2}{2}$
		} &
		\breakcell{Masur--Veech volumes \\ of quadratic differentials} \\
	\midrule
		\breakcell{
			$x(z) = \alpha + \gamma(z + z^{-1})$, \\[1ex]
			$y(z) = \sum_{k=1}^{d-1} u_k z^{-k}$
		} &
		Maps \\
	\midrule
		$x(z) = z + z^{-1}, \quad y(z) = \log(z)$ &
		\breakcell{stationary Gromov--Witten \\ invariants of $\P^1$} \\
	\bottomrule
\end{tabular}
\caption{
	Examples of spectral curves with underlying Riemann surface (open subsets of) $\Sigma = \C$ and bi-differential $\omega_{0,2}(z_1,z_2) = \frac{\dd z_1 \dd z_2}{(z_1 - z_2)^2}$, unless otherwise stated. On the right, the geometric interpretation of the associated amplitudes or generalised periods.
}
\label{tbl:exmpls}
\end{table}

\subsection{Proof strategy}
Our starting point is the expansion of the topological recursion correlators on a natural basis of meromorphic differential $1$-forms:
\begin{equation}
	\omega_{g,n}(z_1,\dots,z_n)
	=
	\sum_{\substack{
		k_1,\dots,k_n \geq 0 \\
		a_1,\dots,a_n \in \mf{a} \\
		k_1+\cdots+k_n \leq 3g-3+n
	}}
		F_{g;(a_1,k_1), \dots, (a_n,k_n)}
		\prod_{i=1}^n \xi^{(a_i,k_i)}(z_i) \,.
\end{equation}
The scalars $F_{g;(a_1,k_1),\dots,(a_n,k_n)}$, called \emph{amplitudes} or $n$-point correlators, satisfy a recursion that takes the form of an Airy structure and corresponds to Virasoro constraints. The amplitudes can be explicitly computed from the germs of $x, y, \omega_{0,2}$ near the ramification points (\cref{sec:TR:basis}). In order to bound such amplitudes, we rely on exponential bounds on the expansion coefficients of $x, y, \omega_{0,2}$. This ``boundedness'' property is identified in \cref{def:boundedness}; it is automatically satisfied by spectral curves, but is an additional assumption if one works with local spectral curves (i.e. $\Sigma$ is only a formal neighbourhood of $\mathfrak{a}$). It holds, for instance, for local spectral curves associated with semi-simple points of Frobenius manifolds.

A central role in our analysis is played by the Painlev\'e~I spectral curve, whose associated amplitudes have a simple exponential dependence on the indices. For an arbitrary spectral curve or a bounded local spectral curve, the general amplitudes are bounded from above by exponentials, thus implying an upper bound by the Painlev\'e~I amplitudes (\cref{sec:comparing}).

By \cite{IZ92} or by the well-known relation between topological recursion and intersection theory on $\overline{\mathcal{M}}_{g,n}$, the Painlev\'e~I amplitudes can be expressed in terms of $\psi$-class intersections (\cref{sec:TR:CohFT}). For the latter, we employed a uniform upper bound provided by Aggarwal which implies a $\frac{(3g - 3 + 2n)!}{g!n!}$-type bound for the Painlev\'e~I amplitudes. For fixed $n$, this can be recast as a $\Gamma(2g - 2 + n)$ bound. The proof can be adapted to obtain the same type of bounds for geometric series of the form $\sum_{k_1, \dots, k_n \geq 0} v^{k_1 + \cdots + k_n} F_{g;k_1,\dots,k_n}^{\text{PI}}$. The evaluation of correlators in local coordinates or the computation of generalised periods can be bounded from above by analogous geometric series. This yields the final results of \cref{sec:proof}.

\subsection{Comments}
A similar proof strategy can be implemented to obtain a lower bound. More precisely, for spectral curves producing non-negative amplitudes, we can provide a lower bound using the Airy or Painlev\'e~I amplitudes (cf. \cref{subsec:lower:bound}). This gives $(2g)!$ lower bounds, but with much less uniformity and an exponential growth rate that does not match the upper bound.

Regarding the optimality of our upper bounds, we have not attempted (and it seems rather difficult with our method) to derive upper bounds with a completely optimal exponential growth rate (the constants $\ms{A}$ in \cref{sec:proof}). Power-law factors like $g^{\varkappa}$ have been ignored, sometimes at the cost of increasing the exponential growth rate and the constant prefactor. Besides, when there are several ramification points, the bounds we use are rather crude and the method may be improved by analysing directly the sum over stable graphs that expresses $F_{g;\alpha_1,\ldots,\alpha_n}$ in terms of amplitudes of spectral curves with a single ramification point. Accordingly, the constant prefactor we obtain (the constant $\ms{S}$ in \cref{sec:proof}) has little significance. Nonetheless, the value of the constant $\ms{A}$ provides a lower bound for the absolute value of the closest singularity to the origin in the Borel plane of the corresponding $n$-point function.

It is worth mentioning that, if an optimal upper bound for the Painlev\'e~I general $n$-point correlators were known, this would slightly improve our \cref{lemma:bound:PI,lem:bound:sums:PI} and thus the upper bounds for general spectral curves. Currently, this is only known in the form of an asymptotic equivalence for the Painlev\'e~I free energies (i.e. $0$-point correlators) due to the work of Kapaev \cite{Kap04} for the leading asymptotic expansion, with subleading asymptotic transseries addressed in \cite{GIKM12,ASV12,BSSV23}. We show in \cref{sec:PIfree} that the upper bound we obtained is not far from the asymptotic equivalent in this case.

To conclude, our method could be adapted to treat irregular spectral curves if one had a sufficiently good uniform upper bound for the intersection numbers
\begin{equation}
	\int_{\overline{\mc{M}}_{g,n}} \Theta_{g,n} \, \psi_1^{k_1} \cdots \psi_n^{k_n}
\end{equation}
where $\Theta_{g,n}$ is Norbury's class \cite{Nor23,CGG}. This would be the analogue for the $\Theta$ class of the uniform bound \eqref{eq:Agg:upper} proved in \cite[Proposition 1.2]{Agg21}. However, new techniques would be needed to handle spectral curves with non-simple ramifications, such as intersections with Witten $r$-spin. Indeed, in this case, we do not know how to upper bound the associated amplitudes (corresponding to W-constraints rather than Virasoro constraints) by those of reference spectral curves, as the latter do not have constant sign. Note that the exact asymptotics of $\Theta$-class or Witten $r$-spin class intersections with $\psi$-classes found in \cite{EGGGL} would not be sufficient in either case, since we require uniform upper bounds instead.

\subsection*{Acknowledgements}
We would like to thank Veronica~Fantini, Stavros~Garoufalidis and Mingkun~Liu for helpful discussions and Ricardo~Schiappa for pointing out important references. This work was initiated when G.B. and A.G. were affiliated with the Max-Planck-Institut f\"ur Mathematik, Bonn, and benefited at this time from the support of the Max-Planck-Gesellschaft. It was completed during a stay of G.B. at IH\'ES, which he thanks for the excellent working conditions. A.G. was supported by an ETH Fellowship (22-2~FEL-003) and a Hermann-Weyl-Instructorship from the Forschungsinstitut für Mathematik at ETH Zürich. It has been supported in part by the ERC-SyG project ``Recursive and Exact New Quantum Theory'' (ReNewQuantum) which received funding from the European Research Council (ERC) under the European Union's Horizon 2020 research and innovation programme under grant agreement No 810573.

\section{Topological recursion: correlators and amplitudes}
\label{sec:TR:basis}

\subsection{Spectral curves and topological recursion}
\label{subsec:TR:SC}
\begin{definition}[\cite{EO07}]\label{def:spectral:curve}
	For us, a \emph{spectral curve} is a quadruple $\mc{S} = (\Sigma,x,y,\omega_{0,2})$ consisting of
	\begin{itemize}
		\item a smooth complex curve $\Sigma$;

		\item two (possibly multi-valued and ill-defined at isolated points) functions $x$ and $y$ such that $\dd x$ and $\dd y$ are meromorphic;

		\item a symmetric bidifferential $\omega_{0,2}$ on $\Sigma \times \Sigma$, whose only singularity is a double pole on the diagonal with biresidue $1$.
	\end{itemize}
	We denote by $\mf{a}$ the set of zeroes of $\dd x$, which we assume to be finite. We also assume that all zeroes are simple. A particularly nice class of spectral curves is that of \emph{regular} spectral curves, for which $y$ is holomorphic at every $a \in \mf{a}$ and $\dd y(a) \neq 0$.
\end{definition}

The assumptions on the functions $x$ and $y$ cover for instance the case they have logarithmic singularities. We do not consider the case of log-vital singularities treated in \cite{ABDKS24}. We also do not consider the case of irregular spectral curves \cite{DN18}, nor ramification points of higher order \cite{BHLMR14,BE13}.
	
Given a regular spectral curve, the \emph{topological recursion} produces a sequence of meromorphic symmetric multidifferentials $\omega_{g,n}$ on $\Sigma^n$, indexed by $(g,n) \in \Z_{\geq 0} \times \Z_{> 0}$ and called \emph{correlators}.
	
\begin{definition}
	We set $\omega_{0,1} = y \dd x$, while $\omega_{0,2}$ is part of the datum of a spectral curve, and the remaining correlators are defined by induction on $2g-2+n > 0$ via the formula
	\begin{multline}\label{eq:TR}
		\omega_{g,n}(z_1,\dots,z_n) \coloneqq \sum_{a \in \mf{a}} \Res_{z = a} \ K_a(z_1,z) \bigg(
			\omega_{g-1,n+1}(z,\sigma_a(z),z_2,\dots,z_n) \\
			+
			\sum_{\substack{g_1+g_2 = g \\ J_1 \sqcup J_2 = \{2,\dots,n\}}}^{\text{no $(0,1)$}}
				\omega_{g_1,1+|J_1|}(z,z_{J_1})
				\omega_{g_2,1+|J_2|}(\sigma_a(z),z_{J_2})\,
		\bigg),
	\end{multline}
	The $K_a$ are called \emph{recursion kernels}, they are locally defined in a neighbourhood $U_a$ of $a \in \mf{a}$ as
	\begin{equation}\label{eq:TR:kernel}
		K_a(z_1,z) \coloneqq \frac{\frac{1}{2} \int_{\sigma_a(z)}^z \omega_{0,2}(z_1,\cdot)}{\bigl( y(z) - y(\sigma_a(z)) \bigr) \dd x(z)}\,,
	\end{equation}
	where $\sigma_a \colon U_a \to U_a$ is the Galois involution near the simple ramification point $a \in U_a$, i.e. the non-trivial holomorphic map such that $x \circ \sigma_a = x$.
\end{definition}

It can be shown that $\omega_{g,n}$ is symmetric in its $n$ variables with vanishing residues, and for $2g - 2 + n > 0$ it has poles only at the ramification points of order at most $6g-4+2n$. In other words:
\begin{equation}
	\omega_{g,n} \in H^0 \Bigl( \Sigma^n,K_{\Sigma}\bigl( (6g-4+2n)\mf{a} \bigr)^{\boxtimes n} \Bigr)^{\mf{S}_n} \,,
\end{equation}
where $K_{\Sigma}$ is the canonical divisor and $\mf{a}$ is interpreted as the ramification divisor of the meromorphic function $x$. Besides, the correlators have no residues at the ramification points.

It is also possible to define correlators for $n = 0$, which are usually denoted as $F_g = \omega_{g,0}$.

\begin{definition}\label{def:Fg}
	We define a collection of scalars $(F_{g})_{g \geq 2}$, called \emph{free energies}, by the formula
	\begin{equation}\label{eq:Fg}
		F_g
		=
		\omega_{g,0}
		\coloneqq
		\frac{1}{2 - 2g} \sum_{a \in \mf{a}} \Res_{z = a} \left( \int_{a}^{z} y \, \dd x \right) \omega_{g,1}(z) \,.
	\end{equation}
	We will not consider $F_0$ and $F_1$ here.
\end{definition}

The free energies are examples of \emph{generalised periods}. In applications of topological recursion to enumerative geometry, the enumerative information is often stored in generalised periods. These are obtained by application of a multilinear form
\begin{equation}
	I_1 \otimes \cdots \otimes I_n
	\in
	\left( H^0(\Sigma,K_{\Sigma}(*\mf{a}))^* \right)^{\otimes n}
\end{equation}
to the correlators $\omega_{g,n}$. Besides free energies, which correspond to the linear form
\begin{equation}
	\mc{F} \colon
	\omega
	\longmapsto
	\sum_{a \in \mf{a}} \Res_{z = a}
		\left( \int_{a}^{z} y \, \dd x \right) \omega(z) \,,
\end{equation}
other prominent examples of generalised periods include monomials in the following linear operators.
\begin{itemize}
	\item
	The expansion coefficients near a point $p \in \Sigma \setminus\mf{a}$. These are obtained by applying the linear form
	\begin{equation}\label{eq:gen:period}
		I_{(p,k)} \colon
		\omega
		\longmapsto
		\Res_{z = p} \ \frac{\omega(z)}{X^{k}(z)} \,.
	\end{equation}
	Here $X$ is a local coordinate centred at $p$ such that
	\begin{equation}\label{eq:coord:X}
		\dd x =
		\begin{cases}
			X^{-d_p} \dd X & d_p \neq 1 \\
			c_p \frac{\dd X}{X} & d_p = 1
		\end{cases} \,,
		\quad\text{ where }\quad
		d_p \coloneqq -\ord_p \dd x
		\text{ and }
		c_p \coloneqq \Res_p \ \dd x \,.
	\end{equation}

	\item
	The evaluation in a local coordinate $Z$ around a point $p \in \Sigma \setminus \mf{a}$, i.e.
	\begin{equation}\label{eq:gen:period:ev}
		\mr{ev}_{(Z,p)} \colon
		\omega
		\longmapsto
		\Res_{z = p} \ \frac{\omega(z)}{Z(z) - Z(p)}
		\eqqcolon
		\frac{\omega(z)}{\dd Z(z)}\bigg|_{z = p}\,,
	\end{equation}

	\item
	A generalisation of both $I_{(p,k)}$ and the evaluation $\mr{ev}_{(Z,p)}$ is
	\begin{equation}\label{eq:gen:period2}
		I_{(Z,p,k)} \colon
		\omega
		\longmapsto
		\Res_{z = p} \ \frac{\omega(z)}{(Z(z) - Z(p))^{k}} \,,
	\end{equation}
	where $Z$ is a local coordinate around a point $p \in \Sigma \setminus \mf{a}$. It extracts the $(k - 1)$-th coefficient of expansion in the local coordinate $Z$.

	\item
	The integration along a path $\gamma$ in $\Sigma \setminus \mf{a}$ between two points:
	\begin{equation}\label{eq:gen:period:int}
		\int_{\gamma} \colon
		\omega
		\longmapsto
		\int_{\gamma} \omega \,.
	\end{equation}
\end{itemize}
The linear operator $I_{(p,k)}$ appears naturally in various applications, such as Hurwitz theory. The integration along a path is used to define the wave function.

\begin{definition}\label{def:wave}
	Given a base point $z_0 \in \Sigma \setminus \mf{a}$ and $\chi \geq -1$, we introduce the (possibly multivalued) meromorphic functions of $z \in \Sigma$:
	\begin{equation}\label{eq:S:chi}
	\begin{split}
		f_{z_0,-1}(z)
		&\coloneqq
			\int_{z_0}^z y \, \dd x \,, \\
		f_{z_0,0}(z)
		&\coloneqq
			\frac{1}{2} \int_{w_1 = z_0}^z \int_{w_2 = z_0}^z \left(
				\omega_{0,2}(w_1,w_2) - \frac{\dd x(w_1)\dd x(w_2)}{(x(w_1) - x(w_2))^2}
			\right) \,,\\
		f_{z_0,\chi}(z)
		&\coloneqq
		\sum_{\substack{g \geq 0, \ n > 0 \\ 2g - 2 + n = \chi}} \frac{1}{n!} \underbrace{\int_{z_0}^{z} \cdots \int_{z_0}^z}_{n \ \text{times}} \omega_{g,n} \,,
		\qquad
		\text{for $\chi > 0$.}
	\end{split}
	\end{equation}
	In the formula for $f_{z_0,-1}$, a suitable regularisation of $\int_{z_0}^{z} y \, \dd x$ should be used if one chooses $z_0$ to be a non-integrable singularity of $y\dd x$. The \emph{wave function} is the formal series of exponential type in $\hbar$ defined as
	\begin{equation}\label{eq:wave:fnct}
		\psi_{z_0}(z;\hbar) \coloneqq \exp{\Biggl( \sum_{\chi \geq -1} \hbar^{\chi} \, f_{z_0,\chi}(z) \Biggr)} \,.
	\end{equation}
\end{definition} 

Equivalently, since the correlators for $2g - 2 + n > 0$ have no residues, the expressions in \eqref{eq:S:chi} can be considered as meromorphic functions of points $z_0,z$ in the universal cover of $\Sigma$. If $\pi_1(\Sigma) = \set{\id}$ (for instance, if $\Sigma = \P^1)$ the functions $f_{z_0,\chi}$ for $\chi \geq 1$ are actually meromorphic functions of $z_0,z \in \Sigma$. If $\pi_1(\Sigma) \neq \set{0}$, they are multivalued and in fact \eqref{eq:wave:fnct} is not the `right' function to look at. If $\Sigma$ is a Riemann surface of positive genus, a better-behaved quantity is the `non-perturbative wave function' \cite{EM11,BE12,EGMO24}. Discussing it is beyond the scope of this article.

\subsection{Quantum Airy structures}
\label{subsec:TR:QAS}
The approach to topological recursion proposed by Kontsevich and Soibelman \cite{KS18} (see also \cite{ABCO24}) starts from a collection of at most quadratic differential operators $(L_{\alpha})_{\alpha \in \mf{A}}$ forming a Lie subalgebra of the Weyl algebra, and constructs a unique formal power series in $\hbar$ that is annihilated simultaneously by all these differential operators.

\begin{definition}\label{def:QAS}
	Let $V$ be a (possibly infinite-dimensional) vector space over $\C$. Fix a basis $(e_{\alpha})_{\alpha \in \mf{A}}$ of $V$ and let $(x^{\alpha})_{\alpha \in \mf{A}}$ be the dual basis. Define the \emph{Weyl algebra} as
	\begin{equation}
		\mc{W}_{\hbar}(V)
		\coloneqq \C[\hbar]\braket{(x^{\alpha},\del_{\alpha})_{\alpha \in \mf{A}}} / \braket{[\del_{\alpha},x^{\beta}] = \hbar\delta_{\alpha}^{\beta}}\,.
	\end{equation}
	A \emph{quantum Airy structure} on $V$ is a collection $(L_{\alpha})_{\alpha \in \mf{A}}$ of elements of $\mc{W}_{\hbar}(V)$ of the form
	\begin{equation}\label{eq:QAS:L}
		L_{\alpha}
		=
		\hbar \del_{\alpha} - \sum_{\mu,\nu \in \mf{A}} \left(
			\frac{1}{2} A_{\alpha,\mu,\nu} x^{\mu} x^{\nu}
			+
			B^{\nu}_{\alpha,\mu} x^{\mu} \del_{\nu}
			+
			\frac{1}{2} C^{\mu,\nu}_{\alpha} \del_{\mu} \del_{\nu}
		\right)
		-
		\hbar D_{\alpha}
	\end{equation}
	that forms a Lie subalgebra of $\mc{W}_{\hbar}(V)$, that is
	\begin{equation}
		\bigl[ L_{\alpha}, L_{\beta} \bigr]
		=
		\hbar \sum_{\mu \in \mf{A}} f^{\mu}_{\alpha,\beta} L_{\mu}\,.
	\end{equation}
	for some scalars $f^{\mu}_{\alpha,\beta}$. In this definition, we can always assume that $A_{\alpha,\beta,\gamma} = A_{\alpha,\gamma,\beta}$ and $C^{\beta,\gamma}_{\alpha} = C^{\gamma,\beta}_{\alpha}$. Moreover, if $V$ is infinite-dimensional, the coefficients have to satisfy some vanishing properties for the next result to make sense and be valid. For all the applications we have in mind, it suffices to require that the set $\mf{A}$ comes with a surjection with finite fibres $N \colon \mf{A} \rightarrow \Z_{\geq 0}$ and $d \in \Z_{\geq 0}$ such that
	\begin{equation}\label{eq:vanish:N}
	\begin{split}
		N(\alpha) + N(\beta) + N(\gamma) > d &\qquad \Longrightarrow \qquad A_{\alpha,\beta,\gamma} = 0\,, \\
		N(\alpha) + N(\beta) - N(\gamma) > d & \qquad \Longrightarrow \qquad B_{\alpha,\beta}^{\gamma} = 0\,, \\
		N(\alpha) - N(\beta) - N(\gamma) > d & \qquad \Longrightarrow \qquad C_{\alpha}^{\beta,\gamma} = 0\,, \\
		N(\alpha) > d & \qquad \Longrightarrow \qquad D_{\alpha} = 0\,.
	\end{split}
	\end{equation}
\end{definition}

As stressed above, the main feature of quantum Airy structures is the existence of a unique partition function annihilated by the operators $(L_{\alpha})_{\alpha \in \mf{A}}$.

\begin{theorem}[{\cite{KS18}}] \label{thm:KS}
	There exists a unique formal series in $\hbar$ of exponential type of the form
	\begin{equation}
		Z(\bm{x};\hbar) \coloneqq \exp\left(
		\sum_{\substack{g \geq 0 , \, n > 0 \\ 2g-2+n > 0}} \frac{\hbar^{g-1}}{n!} \sum_{\alpha_1,\dots,\alpha_n \in \mf{A}}
			F_{g;\alpha_1,\dots,\alpha_n} x^{\alpha_1} \cdots x^{\alpha_n}
		\right)
	\end{equation}
	such that the scalars $F_{g;\alpha_1,\dots,\alpha_n}$, called \emph{amplitudes}, are symmetric under the permutation of the indices $\alpha_1,\dots,\alpha_n$ and
	\begin{equation}
		\forall \alpha \in \mf{A}\qquad	L_{\alpha} \cdot Z(\bm{x};\hbar) = 0\,.
	\end{equation}
	Moreover, the amplitudes are uniquely determined by the following recursion on $2g - 2 + n > 0$:
	\begin{equation}\label{eq:TR:QAS}
	\begin{split}
		F_{g;\alpha_1,\dots,\alpha_n}
		& =
		\sum_{m=2}^n \sum_{\mu \in \mf{A}}
			B_{\alpha_1,\alpha_m}^{\mu} F_{g;\mu,\alpha_2,\dots,\widehat{\alpha_m},\dots,\alpha_n} \\
		&\quad +
		\frac{1}{2} \sum_{\mu,\nu \in \mf{A}}
			C_{\alpha_1}^{\mu,\nu} \Biggl(
				F_{g-1;\mu,\nu,\alpha_2\dots,\alpha_n}
				+
				\!\!\!
				\sum_{\substack{g_1 + g_2 = g \\ I_1 \sqcup I_2 = \{\alpha_2,\dots,\alpha_n\}}}
				\!\!\!
				F_{g_1;\mu,I_1} \, F_{g_2;\nu,I_2}
			\Biggr)\,,
	\end{split}
	\end{equation}
	together with the initial conditions $F_{0;\alpha} = F_{0;\alpha,\beta} = 0$, $F_{0;\alpha,\beta,\gamma} = A_{\alpha,\beta,\gamma}$, and $F_{1;\alpha} = D_{\alpha}$. In the infinite-dimensional setting of \cref{def:QAS}, one can show that for each $(g,n)$ there exists $d_{g,n}$ such that
	\begin{equation}
		N(\alpha_1) + \cdots + N(\alpha_n) > d_{g,n} \qquad \Longrightarrow \qquad F_{g;\alpha_1,\dots,\alpha_n} = 0\,.
	\end{equation}
	In other words, for a fixed $(g,n)$ there are only finitely many non-vanishing amplitudes.
\end{theorem}

Consider now a regular spectral curve $\mc{S}$. Following \cite{ABCO24} (see also \cite{KS18}) we describe a quantum Airy structure whose amplitudes give the decomposition coefficients of the topological recursion correlators on a suitable basis of differentials. The characterisation of these decomposition coefficients by \cref{thm:KS} is then equivalent to Virasoro constraints.

Let $\mf{A} \coloneqq \mf{a} \times \Z_{\geq 0}$. Choose local coordinates $\zeta_a$ on a neighbourhood $U_a \subset \Sigma$ of $\mf{a}$ such that
\begin{equation}
	\zeta_a(a) = 0\qquad \text{and} \qquad x(z) = x(a) + \zeta_a(z)^2\,.
\end{equation}
We consider the meromorphic differentials $(\xi^{\alpha})_{\alpha \in \mf{A}}$ on $\Sigma$ defined for $\alpha = (a,i)$ by
\begin{equation}\label{eq:basis:xi}
	\xi^{(a,i)}(z)
	\coloneqq
	\Res_{w = a} \left(
		\int_{a}^{w} \omega_{0,2}(z, \,\cdot\,)
	\right) \frac{(2i+1)!!}{(\zeta_a(w))^{2i+2}} \dd\zeta_a(w) \,.
\end{equation}
Define also in a small enough neighbourhood $U_a$ of $a$ the meromorphic functions and the inverse of a $1$-form
\begin{equation}
	e_{(a,i)}(z)
	\coloneqq
	\frac{\zeta_a(z)^{2i+1}}{(2i+1)!!}\,,
	\qquad\qquad
	\theta_a(z)
	\coloneqq
	\frac{-2}{\bigl( y(z) - y(\sigma_a(z)) \bigr) \dd x(z)}\,.
\end{equation}
We extend $e_{(a,i)}(z)$ to zero on $U_{b}$ for $b \ne a$. We can then define the tensors
\begin{equation}\label{eq:ABCD:from:TR}
	\begin{split}
		A_{\alpha,\beta,\gamma}
			& \coloneqq
			\Res_{z = a} \ \theta_a(z) \, e_{\alpha}(z) \, \dd e_{\beta}(z) \, \dd e_{\gamma}(z)\,,
			\\
		B_{\alpha,\beta}^{\gamma}
			& \coloneqq
			\Res_{z = a} \ \theta_a(z) \, e_{\alpha}(z) \, \dd e_{\beta}(z) \, \xi^{\gamma}(z)\,,
			\\
		C_{\alpha}^{\beta,\gamma}
			& \coloneqq
			\Res_{z = a} \ \theta_a(z) \, e_{\alpha}(z) \, \xi^{\beta}(z) \, \xi^{\gamma}(z)\,,
			\\
		D_{\alpha}
			& \coloneqq
			\delta_{i,0} \left(
				\frac{\theta_{(a,0)}}{2} \phi_{(a,0),(a,0)} + \frac{\theta_{(a,1)}}{8}
			\right)
			+
			\delta_{i,1} \frac{\theta_{(a,0)}}{24}\,,
	\end{split}
\end{equation}
for $\alpha = (a,i)$. In the definition of $D_{\alpha}$, the coefficients are given by the expansion of $\theta_a$ and $\omega_{0,2}$ as
\begin{equation}\label{eq:expns:coeff:y:B}
\begin{split}
	\theta_a(z)
	& \underset{z \to a}{\approx}
	\Biggl( \sum_{m \geq 0} \theta_{(a,m)} \, \zeta_{a}(z)^{2m - 2} \Biggr) \frac{1}{\dd\zeta_a(z)} \,, \\
	\omega_{0,2}(z_1,z_2)
	& \underset{\substack{z_1 \to a_1 \\ z_2 \to a_2}}{\approx}
		\Bigg(
			\frac{\delta_{a_1,a_2}}{(\zeta_{a_1}(z_1) - \zeta_{a_2}(z_2))^2} + \text{non-even holomorphic} \\
	& \qquad\quad
			+
			\sum_{m_1,m_2 \geq 0} \phi_{(a_1,m_1),(a_2,m_2)}
				\zeta_{a_1}(z_1)^{2m_1} \, \zeta_{a_2}(z_2)^{2m_2} 
		\Bigg) \dd\zeta_{a_1}(z_1) \, \dd\zeta_{a_2}(z_2)\,.
\end{split}
\end{equation}
The symbol $\approx$ means an all-order equality of the series expansion in the specified regime. In the last formula, `non-even holomorphic' refer to terms with non-negative powers of $\zeta_{a_1}(z_1)$ or $\zeta_{a_2}(z_2)$, at least one of them being odd. One can check that with the natural projection $N \colon \mf{A} \rightarrow \Z_{\geq 0}$ and choosing $d = 2$, the vanishing conditions \eqref{eq:vanish:N} are met.
 
\begin{proposition}[{\cite{ABCO24}}]\label{prop:ABCDWGN}
	The tensors $(A,B,C,D)$ defined by \cref{eq:ABCD:from:TR} form a quantum Airy structure on the vector space with basis $(e_{\alpha})_{\alpha \in \mf{A}}$. Besides, the topological recursion correlators satisfy
	\begin{equation}\label{eq:EO:KS:TR}
		\omega_{g,n}(z_1,\dots,z_n)
		=
		\sum_{\alpha_1,\dots,\alpha_n \in \mf{A}}
			F_{g;\alpha_1,\dots,\alpha_n} \prod_{i=1}^n \xi^{\alpha_i}(z_i)\,,
	\end{equation}
	and the sum is finite: $F_{g;\alpha_1,\dots,\alpha_n}$ vanishes if $k_1 + \cdots + k_n > d_{g,n} \coloneqq 3g-3+n$, with $\alpha_i = (a_i,k_i)$.
\end{proposition}

The quantum Airy structure \eqref{eq:ABCD:from:TR} has an explicit presentation in terms of the expansion coefficients \eqref{eq:expns:coeff:y:B}. Indeed, the expansion of the differentials $\xi^{(a,i)}(z)$ around an arbitrary ramification point $z = b$ is given by (cf. \cite[Equation~(4.8)]{Eyn14})
 \begin{equation}\label{eq:xi:expansion}
	\xi^{(a,i)}(z)
	\underset{z \to b}{\approx}
	\Biggl(
		\delta_{a,b} \frac{(2i+1)!!}{\zeta_{b}(z)^{2i+2}}
		+
		(2i-1)!! \sum_{j \geq 0} \phi_{(a,i),(b,j)} \zeta_{b}(z)^{2j} + \text{odd holomorphic}
	\Biggr) \dd\zeta_{b}(z)\,,
\end{equation}
with the convention that $(-1)!! = 1$. Moreover, by direct computation, one finds the following expression for the quantum Airy structure, expressed in terms of the multi-index Kronecker delta: $\delta_{\alpha_1,\dots,\alpha_k}$ equals $1$ if $\alpha_1 = \cdots = \alpha_k$ and $0$ otherwise.

\begin{proposition}\label{prop:ABCD:explicit}
	The quantum Airy structure defined by \cref{eq:ABCD:from:TR} is given in terms of the expansion coefficients of $(\theta_\alpha)_{\alpha}$ and $(\phi_{\alpha,\beta})_{\alpha,\beta}$ by the following formulae:
	\begingroup
	\allowdisplaybreaks
	\begin{align}\label{eq:ABCD:explicit}
		\notag
		A_{\alpha,\beta,\gamma}
			&=
			\delta_{a,b,c} \, \delta_{i,j,k,0} \,
				\theta_{(a,0)}\,,
			\\
		\notag
		B_{\alpha,\beta}^{\gamma}
			&=
			\delta_{a,b,c} \,
				\theta_{(a,k-i-j+1)} \,
				\frac{(2k+1)!!}{(2i+1)!! (2j-1)!!} 
			+
			\delta_{a,b} \, \delta_{i,j,0} \,
				\theta_{(a,0)} \, \phi_{(a,0),(c,k)} (2k-1)!!\,,
			\\
		C_{\alpha}^{\beta,\gamma}
			&=
			\delta_{a,b,c} \,
				\theta_{(a,j+k-i+2)} \,
				\frac{(2j+1)!! (2k+1)!!}{(2i+1)!!} \\\notag
			&\qquad
			+
			\delta_{a,b} 
				\sum_{\substack{m,m' \geq 0 \\ m+m' = j-i+1}}
					\theta_{(a,m)} \, \phi_{(a,m'),(c,k)} \frac{(2j+1)!! (2k-1)!!}{(2i+1)!!} \\\notag
			&\qquad
			+
			\delta_{a,c} 
				\sum_{\substack{m,m' \geq 0 \\ m+m' = k-i+1}}
					\theta_{(a,m)} \, \phi_{(a,m'),(b,j)} \frac{(2j-1)!! (2k+1)!!}{(2i+1)!!} \\\notag
			&\qquad
			+
			\delta_{i,0} \,
			\theta_{(a,0)} \, \phi_{(a,0),(b,j)} \, \phi_{(a,0),(c,k)} \,
			(2j-1)!! (2k-1)!!\,,
		\\
		\notag
		D_{\alpha}
			&=
			\delta_{i,0} \left(
				\frac{\theta_{(a,0)}}{2} \phi_{(a,0),(a,0)} + \frac{\theta_{(a,1)}}{8}
			\right)
			+
			\delta_{i,1} \frac{\theta_{(a,0)}}{24}\,,
	\end{align}
	\endgroup
	where $\alpha = (a,i)$, $\beta = (b,j)$, $\gamma = (c,k)$. By convention $\theta_{(a,m)} = 0$ for $m < 0$ and empty sums are equal to zero.
\end{proposition}

Although not mentioned in \cite{ABCO24}, \cref{prop:ABCDWGN} can also be extended to describe the topological recursion free energies. More precisely, they are naturally computed in terms of the expansion coefficients of $\omega_{0,1} = y \, \dd x$, namely
\begin{equation}\label{eq:ydx:expns}
	\frac{1}{\theta_a(z)}
	=
	-
	\frac{y(z) - y(\sigma_a(z))}{2} \dd x(z)
	\mathop{\approx}_{z \rightarrow a}
	\Biggl( \sum_{i \geq 1} t_{(a,i)} \zeta_a(z)^{2i}\Biggr)\dd \zeta_a(z)\,.
\end{equation}
Note that for each $a \in \mf{a}$ and $i \geq 1$, the coefficient $t_{(a,i)}$ is a Laurent polynomial in $\theta_{(a,0)}$ and a polynomial in $(\theta_{(a,j)})_{j \geq 1}$. Conversely, for each $a \in \mf{a}$ and $i \geq 0$, the coefficient $\theta_{(a,i)}$ is a Laurent polynomial in $t_{(a,1)}$ and a polynomial in $(t_{(a,j)})_{j \geq 2}$. So, the coefficients $\theta_{\alpha}$ and $t_{\alpha}$ encode the same information. Inserting \eqref{eq:EO:KS:TR} and \eqref{eq:ydx:expns} into the definition \eqref{eq:Fg} of the free energies and using the behaviour of the differentials \eqref{eq:xi:expansion}, we arrive to the following formula.
 
\begin{lemma}\label{lem:Fg}
	For any $g \geq 2$, we have
	\begin{equation}
		F_{g} = \frac{1}{2g - 2} \sum_{(a,i) \in \mf{A}} (2i - 1)!! \, t_{(a,i)} \, F_{g;(a,i)} \,.
	\end{equation}
	By convention, we have set $t_{(a,0)} = 0$.
\end{lemma} 

\subsection{Local spectral curves and Frobenius manifolds}
\label{sec:local:SP}
\emph{Local spectral curves} are defined like spectral curves, except that $\Sigma$ is taken to be a finite disjoint union of formal disks $\Sigma_a = \Spec \, \C\bbraket{\zeta_a}$ for $a \in \mf{a}$ centred at the zeros of $\dd x$. The data of a local spectral curve is therefore fully determined by the expansion \eqref{eq:expns:coeff:y:B} considered as formal series. All the definitions and properties in \cref{subsec:TR:SC,subsec:TR:QAS} continue to make sense in the weaker setting of local spectral curves (except that the generalised periods have a limited interest). The quantum Airy structure amplitudes then only depend on the expansion coefficients $t_{\alpha}$ (or $\theta_{\alpha}$) and $\phi_{\alpha,\beta}$, while the topological recursion correlators also depend (via the differentials $\xi^{\alpha}$) on the 'non-even holomorphic' part of the expansion of $\omega_{0,2}$. Spectral curves determine local spectral curves by keeping only the germ at $\mf{a}$ of $x,y$ and $\omega_{0,2}$, but local spectral curves do not necessarily come from spectral curves as the formal series may have zero radius of convergence.

Local spectral curves naturally appear in relation to the ancestor potential of semi-simple Frobenius manifolds, which is computed by topological recursion on a local spectral curve \cite{DOSS14}. This correspondence will be recalled in \cref{subsec:Mgn} after we introduce the necessary background concerning the moduli space of curves. For the moment, we simply review the construction of the local spectral curve from the Frobenius manifold. This relies on the construction of a basis of flat sections for Dubrovin's extended connection, and we refer to \cite[Lecture~3]{Dub96} for details.

Consider a \emph{Frobenius manifold} $\mf{X}$, with (complex, non-degenerate, flat) metric $\eta$, Levi-Civita connection $\nabla$, product $\star$ on tangent spaces, Euler vector field $E$, and unit vector field $\mb{1}$. Let us define sections $\bm{U},\bm{V}$ of $\End(\mr{T}\mf{X})$ by
\begin{equation}
	\bm{U}(Y) \coloneqq E \star Y\,,
	\qquad\qquad
	\bm{V}(Y) \coloneqq (\id_{\mr{T}\mf{X}} \otimes \eta)(\nabla E \otimes Y)\,.
\end{equation}
Let $\widehat{\mf{X}} \coloneqq \mf{X} \times \P^1$ and $\pi \colon \widehat{\mf{X}} \rightarrow \mf{X}$ the natural projection. We identify sections $X,Y$ of $\pi^*\mr{T}\mf{X}$ with vector fields $X(z),Y(z)$ on $\mf{X}$ that depend on $z \in \P^1$, so that both $\nabla_X Y$ and $\del_{z}Y$ make sense. Then, we have a flat meromorphic connection $\widehat{\nabla}$ on $\pi^*\mr{T}\mf{X}$ defined by
\begin{equation}
	\widehat{\nabla}_{X + f \del_{z}} Y
	\coloneqq
	\left(\nabla_{X} Y + z X\star Y\right) + f \left(\del_{z} - \bm{U}(Y) - z^{-1}\bm{V}(Y)\right) \,,
\end{equation}
where $f$ is a function on $\widehat{\mf{X}}$ and $X,Y$ are sections of $\pi^*\mr{T}\mf{X}$. This connection has a regular singularity at $z = 0$ and an irregular singularity at $z = \infty$, and we are interested in local bases of flat sections.
 
\begin{definition}\label{def:ssF}
	A Frobenius manifold is \emph{semi-simple} at a point $p$ if the tangent space $(\mr{T}_{p}\mf{X},\star|_{p})$ is a semi-simple algebra.
\end{definition}

By the results of \cite{Dub96}, near a semi-simple point $p$ there exist local coordinates $(u_a)_{a \in \mf{a}}$ with $\mf{a} \coloneqq \{1,\ldots,\dim \mf{X}\}$, called canonical, such that $\del_{u_a} \star \del_{u_b} = \delta_{a,b} \del_{u_a}$ for any $a,b \in \mf{a}$, the unit is $\mb{1} = \sum_{a \in \mf{a}} \del_{u_a}$, and $E = \sum_{a \in \mf{a}} u_a \del_{u_a}$. Then $\bm{U}$ acts diagonally with eigenvalues $(u_{a})_{a\in \mf{a}}$. Focusing only on the $\P^1$-component of the connection $\widehat{\nabla}$, there exists a unique formal basis $\Psi(z) = R(1/z)e^{z\bm{U}}$ of solutions of
\begin{equation}
	\del_z \Psi(z) = \bm{U}(\Psi(z)) + z^{-1}\bm{V}(\Psi(z))
\end{equation}
such that
\begin{equation}
	R(\lambda) \in \id_{\mr{T}_p \mf{X}} + \lambda.\End(\mr{T}_p \mf{X})\bbraket{\lambda}
	\qquad\text{and}\qquad
	R(-\lambda)R^{\mr{t}}(\lambda) = \id_{\mr{T}_p \mf{X}}\,.
\end{equation}
We further set $T_R(\lambda) \coloneqq \lambda (\id - R^{-1}(\lambda)) \mb{1}$.

Following \cite{DOSS14} we proceed with the construction of a local spectral curve. We first set $t_{(a,1)} = \sqrt{\eta(\del_{u_a},\del_{u_a})}$ for the norm of $\del_{u_a}$, and let $\del_{v_a} = t_{(a,1)}^{-1} \del_{u_a}$ define the orthonormal canonical vector fields. We take
\begin{equation}
	\Sigma \coloneqq \bigsqcup_{a \in \mf{a}} \Spec \, \C\bbraket{\zeta_a} \,,
	\qquad\qquad
	x(\zeta_a) \coloneqq p_a + \zeta_a^2 \,,
\end{equation}
and set as usual $\mf{A} \coloneqq \mf{a} \times \Z_{\geq 0}$. We specify the coefficients $t_{\alpha}$ and $\phi_{\alpha,\beta}$ for $\alpha,\beta \in \mf{A}$ from the decompositions
\begin{equation}\label{eq:RT:Frobenius}
\begin{aligned}
	&
	\frac{\id_{\mr{T}_p \mf{X}} - R^{-1}(\lambda) \otimes R^{-1}(\mu)}{\lambda + \mu}\eta^{\bot}
	=
	\sum_{(a,i),(b,j) \in \mf{A}}
		(2i - 1)!! (2j - 1)!! \,
		\phi_{(a,i),(b,j)}
		\lambda^i \del_{v_a} \otimes \mu^j \del_{v_b}\,, \\
	&
	T_R(\lambda) - \lambda \mb{1}
	=
	- \sum_{(a,i) \in \mf{A}}
		(2i - 1)!! \,
		t_{(a,i)} \lambda^i \del_{v_a} \,,
\end{aligned}
\end{equation}
where $\eta^{\bot} = \sum_{a \in \mf{a}} \del_{v_a}^{\otimes 2}$ is the co-pairing.

\subsection{Boundedness conditions}
\label{sec:boundedness}
In order to prove the $(2g)!$ factorial growth of topological recursion amplitudes, we need to assume a boundedness property on the associated expansion coefficients $\theta_{\alpha}$ and $\phi_{\alpha,\beta}$. As we will see shortly, boundedness is automatically satisfied by regular spectral curves. However, it is a non-trivial requirement in the local setting. Nonetheless, we will prove that local spectral curves associated with Frobenius manifolds satisfy this condition.

\begin{definition}\label{def:boundedness}
	We say that a (local or not) spectral curve is \emph{bounded} if there exist positive real constants $M_t,M_{\theta},M_{\phi},\rho_t,\rho$ such that, for any $(a,i),(b,j) \in \mf{A}$:
	\begin{equation}\label{eq:bound}
		\big|t_{(a,i)}\big| \leq \frac{M_t}{\rho_t^{i + 1}} \,,
		\qquad
		\big|\theta_{(a,i)}\big| \leq \frac{M_{\theta}}{\rho^{i+1}} \,,
		\qquad 
		\big|\phi_{(a,i),(b,j)}\big| \leq \frac{M_{\phi}}{\rho^{i+j}} \,.
	\end{equation}
	Here $M_{\phi}$ is allowed to be zero.
\end{definition}

The second and third inequalities (involving $\rho$) will be our starting assumptions to derive upper bounds on the topological recursion amplitudes. The first inequality (involving $\rho_t$) will only play a role in the derivation of upper bounds for the topological recursion free energies, as one can expect from \cref{lem:Fg}.

\begin{remark}\label{rem:R:Rt}
	There is no order relation between $\rho_t$ and $\rho$, but the existence of $\rho_t > 0$ for which the first inequality holds is equivalent to the existence of $\rho > 0$ for which the second inequality holds.
\end{remark}

In the remaining part of this section we show that the two main classes of examples of spectral curves considered in this paper, namely regular spectral curves and local spectral curves associated with Frobenius manifolds, are both bounded.

\begin{lemma}\label{lem:reg:SC:bnd}
	Regular spectral curves are bounded.
\end{lemma}

\begin{proof}
	Let $R_a > 0$ be the maximum value for which there exists an open neighbourhood $U_{R_a} \subset \Sigma$ of $a$ with local coordinate $\zeta_a$ that maps biholomorphically to the open disc
	\begin{equation}
		\mb{D}(R_a) = \set{ \zeta \in \C \;|\; |\zeta| < R_a } \,.
	\end{equation}
	If $\dd y - \sigma_a^*\dd y$ has a pole in $U_{R_a}$, we set $R_{a,-}$ to be the distance between $a$ and the closest such pole (measured in $\mb{D}(R_a)$ with the coordinate $\zeta_a$), otherwise we take $R_{a,-} = R_a$. By construction, $(\dd y - \sigma_a^* \dd y)/\dd\zeta_a$ is a holomorphic $1$-form in $U_{R_{a,-}}$, and Cauchy's inequality in the corresponding disc $\mb{D}(R_{a,-})$ shows that we can make $|t_{(a,i)}|$ satisfy the bound \eqref{eq:bound} for any choice of $\rho_t < R_{a,-}^2$. Likewise, if $\dd y - \sigma_a^{*}\dd y$ has a zero in $U_{R_a} \setminus \{a\}$, we set $R_{a,+}$ to be the distance between $a$ and the closest such zero, otherwise we set $R_{a,+} = R_a$, and we infer that we can make $|\theta_{(a,i)}|$ satisfy the bound \eqref{eq:bound} for any choice of $\rho < R_{a,+}^2$. To get the third desired inequality, we remark that
	\begin{equation}
		\frac{\omega_{0,2}(z_1,z_2)}{\dd \zeta_{a_1}(z_1)\dd\zeta_{a_2}(z_2)} - \frac{1}{(\zeta_{a_1}(z_1) - \zeta_{a_2}(z_2))^2}
	\end{equation}
	is a holomorphic function of $(z_1,z_2) \in U_{R_{a_1}} \times U_{R_{a_2}}$. So we can make the corresponding bounds in \eqref{eq:bound} hold for any choice of positive $\rho < \min(R_{a_1}^2,R_{a_2}^2)$. Therefore, we always have all bounds in \eqref{eq:bound} for any choice of positive
	\begin{equation}
		\rho_t < \min_{a \in\mf{a}} R_{a,-}^2
		\qquad\text{and}\qquad
		\rho < \min_{a\in\mf{a}} R_{a,+}^2 \,.
	\end{equation}
\end{proof}

\begin{lemma}\label{lem:ssF:bound}
	The local spectral curve associated with a semi-simple point of a Frobenius manifold is regular and bounded. This property remains true if we replace $T_R(\lambda)$ with some $T(\lambda) \in \lambda^2.\mr{T}_p \mf{X}\bbraket{\lambda}$ obeying 
	\begin{equation}\label{eq:bound:T}
		\max_{a \in \mf{a}} \big|\eta(\del_{v_a},T^{(k)}(0))\big| \leq M^{k + 1} k!^2
	\end{equation}
	for all $k \ge 0$ for some $M > 0$.
\end{lemma}

\begin{proof}
	Regularity is clear. For boundedness, the basic idea is that differential equation with Poincar{\'e} rank $1$ have at most factorially divergent coefficients \cite{Poi86}, and this implies the well-known fact that Dubrovin's superpotential are defined in discs of positive radii \cite{Dub96}. Here we propose an uneducated proof.

	Inserting $\Psi(z) = R(1/z)e^{\bm{U}z}$ with $z = \lambda^{-1}$, we find a differential equation for the $R$-matrix: $R'(\lambda) = \lambda^{-2}[R(\lambda),\bm{U}] - \lambda^{-1} \bm{V}R(\lambda)$. Inserting the expansion $R(\lambda) = \id + \sum_{k \geq 1} R_k\lambda^k$ we get the recursion
	\begin{equation}\label{eq:adjoint}
		\mr{ad}_{\bm{U}} R_{k + 1} = - kR_k - \bm{V}R_k \,.
	\end{equation}
	At a semi-simple point $\bm{U}$ is simple, hence $\mr{ad}_{\bm{U}}$ is invertible. From \eqref{eq:adjoint} we deduce the existence of $M > 0$ (depending on the point on the Frobenius manifold) such that $||R_{k + 1}||_{\infty} \leq (k + 1) M \, ||R_k||_{\infty}$ for any $k \geq 0$. Letting $[R_k]_{a,b} \coloneqq \eta(\del_{v_b},R_k(\del_{v_{a}}))$, this shows $\big|[R_k]_{a,b}\big| \leq M^k \, k!$. We will continue our analysis with a bound of the form
	\begin{equation}\label{eq:Rmat:bound}
		\big|[R_k]_{a,b}\big|
		\leq
		\frac{M_R}{\rho_R^k} \, k!
	\end{equation}
	for some constants $M_R$, $\rho_R > 0$. The advantage of doing so is that, when the $R$-matrix is explicitly known and an optimal bound of the form \eqref{eq:Rmat:bound} can be found, we will be able to track the constants appearing in \cref{def:boundedness}.

	Let us define $\tilde{\phi}_{(a,i),(b,j)} \coloneqq (2i - 1)!!(2j - 1)!! \phi_{(a,i),(b,j)}$, with the convention that it is zero for $i < 0$ or $j < 0$. From the definition of the coefficients $\phi_{(a,i),(b,j)}$ in \eqref{eq:RT:Frobenius} and taking into account $R^{-1}(\lambda) = R^{\mr{t}}(-\lambda)$, we have
	\begin{equation}\label{eq:phi:rel}
		\tilde{\phi}_{(a,i - 1),(b,j)} + \tilde{\phi}_{(a,i),(b,j-1)}
		=
		\delta_{a,b}\delta_{i,j,0} - (-1)^{i+j}\sum_{c \in \mf{a}} [R_i]_{c,a} [R_j]_{c,b}\,.
	\end{equation}
	Call $P_{(a,i),(b,j)}$ the right-hand side of the above equation. Using \eqref{eq:Rmat:bound}, we get the bound
	\begin{equation}
		|P_{(a,i),(b,j)}| \leq \delta_{i,j} + |\mf{a}| \, \rho_R^{-(i + j)} \, M_R^2 \, i!j! \,.
	\end{equation}
	On the other hand, we can consider \eqref{eq:phi:rel} as a recursion in $j$ with initial data $\tilde{\phi}_{\alpha,(b,0)}$. Its solution is given as
	\begin{equation}
	\begin{split}
		\tilde{\phi}_{(a,i),(b,j)}
		& =
		(-1)^{j + 1} \tilde{\phi}_{(a,i + j),(b,0)}
		+
		\sum_{k = 0}^j P_{(a,i + k + 1),(b,j-k)}\,.
	\end{split}
	\end{equation}
	Specialising \eqref{eq:RT:Frobenius} to $\mu = 0$ we also get access to the initial data $\tilde{\phi}_{(a,i),(b,0)} = (-1)^i [R_{i + 1}]_{b,a}$. Combining the bounds for $R$ and $P$ we obtain
	\begin{equation}
		\big|\tilde{\phi}_{(a,i),(b,j)}\big|
		\leq
		M_R \, \rho_R^{-(i + j + 1)} (i+ j + 1)!
		\left(
			1 + |\mf{a}| \, M_R \sum_{k = 0}^j \frac{1}{\binom{i + j + 1}{j - k}}
		\right) .
	\end{equation}
	We bound the second sum by
	\begin{equation}
	\begin{split}
		\sum_{\ell = 0}^{j} \frac{1}{\binom{i + j + 1}{\ell}}
		\leq
		1 + \frac{\delta_{j \geq 1}}{2} + \frac{\delta_{j \geq 2}\delta_{i,0}}{2} + \delta_{j \geq 2}\delta_{i \geq 1} \sum_{\ell = 2}^{i + j - 1} \frac{1}{\binom{i + j + 1}{\ell}}\,.
	\end{split}
	\end{equation}
	The last sum is bounded by $\frac{2(i + j - 2)}{(i + j + 1)(i + j)}$. For $i + j + 3 \geq 0$ it is bounded by $\frac{1}{5}$. Therefore
	\begin{equation}
		\big|\tilde{\phi}_{(a,i),(b,j)}\big|
		\leq
		\frac{M_R(1 + \frac{11}{10} \, |\mf{a}| \, M_R) (i + j + 1)!}{\rho_R^{i + j + 1}} \,.
	\end{equation}
	We now come back to $\phi_{(a,i),(b,j)}$. Consider the inequality
	\begin{equation}
		\frac{(i + j + 1)!}{(2i - 1)!!(2j - 1)!!}
		=
		2^{i+ j} \frac{(i + j + 1)!}{i!j! \binom{2i}{i} \binom{2j}{j}}
		\leq
		2^{i+ j} \, (i + j + 1) \frac{\binom{i + j}{(i + j)/2}}{\binom{2i}{i} \binom{2j}{j}} \,.
	\end{equation}
	The function $x \mapsto 4^{-x}\binom{2x}{x}$ for $x \geq 0$ is decreasing, hence $4^{-x}\binom{2x}{x} \leq 1$ for any $x \geq 0$. The function $x \mapsto 4^{-x}\sqrt{\pi(x + 1/2)}\binom{2x}{x}$ is also decreasing and tends to $1$ as $x \rightarrow \infty$, hence is lower bounded by $1$ for any $x \geq 0$. Therefore
	\begin{equation}
		\frac{(i + j + 1)!}{(2i - 1)!!(2j - 1)!!}
		\leq
		2\pi (i + j + 1)\sqrt{(i + 1/2)(j + 1/2)}
		\leq
		\pi (i + j + 1)^2 \leq \pi \epsilon^{-2} e^{2\epsilon (i + j + 1) - 2}
	\end{equation}
	for any $\epsilon > 0$. Therefore $\big|\phi_{(a,i),(b,j)}\big| \leq M_{\phi} \, \rho_{\phi}^{-i - j}$ where we can take
	\begin{equation}
		\rho_{\phi}
		=
		e^{-2\epsilon} \, \rho_R \,,
		\qquad\qquad
		M_{\phi}
		=
		\frac{\pi e^{2\epsilon}}{(e\epsilon)^2} \frac{M_R (1 + \frac{11}{10} \, |\mf{a}| \, M_R)}{\rho_R} \,.
	\end{equation}
	Furthermore, the choice $T_R(\lambda) = \lambda(\id - R^{-1}(\lambda))\mb{1}$ and the definition of the coefficients $t_{\alpha}$ yield
	\begin{equation}
		(2i - 1)!! t_{(a,i)} = (-1)^i \sum_{c \in \mf{a}} [R_{i - 1}]_{c,a}\,.
	\end{equation}
	Then, for $i \geq 1$, we find $|t_{(a,i)}| \leq \frac{|\mf{a}|M_R}{i \binom{2i}{i} \rho_R} \leq \frac{M_t}{\rho_{t}^{i + 1}}$ with
	\begin{equation}
		\rho_t = 4\rho_R \,,
		\qquad\qquad
		M_t = \sqrt{\pi} \, |\mf{a}| \, M_R \, \rho_R^2 \,.
	\end{equation}
	A similar argument yields the inequality $|t_{(a,i)}| \leq M_t \rho_t^{-i - 1}$ for some $M_t, \rho_t > 0$. The same inequality holds if we replace $T_R(\lambda)$ with any $T(\lambda) \in \mr{O}(\lambda^2)$ satisfying \eqref{eq:bound:T}. By \cref{rem:R:Rt} we deduce the existence of $M_{\theta},\rho_{\theta} > 0$ such that $|\theta_{(a,i)}\big| \leq M_{\theta} \rho_{\theta}^{-i - 1}$. Then the desired inequalities follow, after choosing $\rho = \min \set{\rho_{\theta},\rho_{R}}$.
\end{proof}

\section{Comparing Airy structures and amplitudes}
\label{sec:comparing}

\subsection{Upper bounds}
\label{subsec:upper:bound}
Consider the $1$-parameter family of spectral curves $\mc{S}^{\textup{PI}}(u)$ on $\Sigma = \P^1$, called the \emph{Painlev{\'e} I spectral curve}, given by
\begin{equation}\label{eq:SC:PI}
	x(z) = z^2 - 2u\,,
	\qquad
	y(z) = z^3 - 3uz\,,
	\qquad
	\omega_{0,2}(z_1,z_2) = \frac{\dd z_1 \dd z_2}{(z_1 - z_2)^2}\,.
\end{equation}
It is the spectral curve of the $(3,2)$-minimal model, which occurs in the study of pure $2D$ gravity (see for instance \cite{CED18}). Moreover, its free energies are known to compute the formal asymptotics of the solution of Painlev{\'e} I (PI for short) equation, see \cite{BBE15,IS16}. This paragraph aims at showing that the amplitudes for any bounded spectral curve are dominated by the PI amplitudes (\cref{prop:upper:bound:Fgn}).

To this end, we start by computing the quantum Airy structure associated with $\mc{S}^{\textup{PI}}(u)$. Since the spectral curve has only one ramification point at $a = 0$, we omit its dependence from the associated quantum Airy structure. In other words, we identify $\mf{A} = \set{0} \times \Z_{\ge 0}$ with $\Z_{\ge 0}$. Applying the formulae of \cref{subsec:TR:QAS}, we find the following expressions.

\begin{lemma}\label{lemma:QAS:PI} 
	The quantum Airy structure $(\PIL{A},\PIL{B},\PIL{C},\PIL{D})(u)$ associated with $\mc{S}^{\textup{PI}}(u)$ is given by
	\begin{equation}\label{eq:QAS:PI}
		\begin{split}
			\PIL{A}_{i,j,k}(u)
			&= \frac{\delta_{i,j,k,0}}{2 (3u)}\, ,\\
			\PIL{B}_{i,j}^{k}(u)
			&= \frac{\delta_{k-i-j+1 \geq 0}}{2 (3 u)^{k-i-j+2}} \frac{(2k+1)!! }{(2i+1)!! (2j-1)!!}\, ,\\
			\PIL{C}_{i}^{j,k}(u)
			&= \frac{\delta_{j+k-i+2 \geq 0}}{2 (3 u)^{j+k-i+3}} \frac{(2j+1)!! (2k+1)!!}{(2i+1)!!}\, ,\\
			\PIL{D}_{i}(u)
			&= \frac{1}{2} \left( \frac{\delta_{i,0}}{8 (3u)^2} + \frac{\delta_{i,1}}{24 (3u)} \right) \,.
		\end{split}
	\end{equation}
	All the coefficients are weakly decreasing functions of $u \in \R_{> 0}$.
\end{lemma}

The main interest of these formulae is to allow a direct comparison with the quantum Airy structure of arbitrary spectral curves.

\begin{lemma}\label{lem:upper:bound:ABCD}
	Let $(A,B,C,D)$ be the quantum Airy structure associated with a bounded spectral curve $\mc{S}$. Then for any $\alpha = (a,i),\beta = (b,j),\gamma = (c,k) \in \mf{A}$:
	\begin{equation}
	\begin{split}
		\big|A_{\alpha,\beta,\gamma}\big|
		&\leq
		Q \cdot \PIL{A}_{i,j,k}(u)\,, \\
		\big|B_{\alpha,\beta}^{\gamma}\big|
		&\leq
		Q \cdot \PIL{B}_{i,j}^{k}(u)\,, \\
		\big|C_{\alpha}^{\beta,\gamma}\big|
		&\leq
		Q \cdot \PIL{C}_{i}^{j,k}(u)\,, \\
		|D_{\alpha}|
		&\leq
		Q \cdot \PIL{D}_{i}(u)\,,
	\end{split}
	\end{equation}
	where we can take $Q = 2M_{\theta} \max{ \left\{(1 + \rho M_{\phi})^2,1 + 4\rho M_{\phi}\right\}}$ and $u = 3^{-4/3}\rho$.
\end{lemma}

\begin{proof}
	The tensors $(A,B,C,D)$ were computed in \cref{prop:ABCD:explicit}. For the A-case, we have 
	\begin{equation}
		\big|A_{\alpha,\beta,\gamma}\big|
		=
		\delta_{a,b,c} \, \delta_{i,j,k,0} \, \theta_{(a,0)}
		\le
		\delta_{i,j,k,0} \, \frac{M_{\theta}}{\rho}
		=
		2M_{\theta} \cdot \PIL{A}_{i,j,k}\big(\tfrac{\rho}{3}\big) \,.
	\end{equation}
	Similarly for the B-case, for which we have
	\begin{equation}
	\begin{split}
		\big|B_{\alpha,\beta}^{\gamma}\big|
		& \leq
		\delta_{a,b,c} \, \delta_{k-i-j+1 \geq 0} \,
			\frac{M_{\theta}}{\rho^{k+2-i-j}} \,
			\frac{(2k+1)!!}{(2i+1)!! (2j-1)!!}
			+
		\delta_{a,b} \, \delta_{i,j,0} \,
			\frac{M_{\theta} M_{\phi}}{\rho^{k+1}}
			(2k-1)!! \\
		& \leq
		2M_{\theta} (1 + \rho M_{\phi})
		\cdot
		\PIL{B}_{i,j}^{k}(\tfrac{\rho}{3})\,.
	\end{split}
	\end{equation}
	The C-case is slightly more involved:
	\begin{equation}
	\begin{split}
		\big|C_{\alpha}^{\beta,\gamma}\big|
		& \leq
		\delta_{a,b,c} \, \delta_{j+k-i+2 \geq 0} \,
				\frac{M_{\theta}}{\rho^{j+k-i+3}} \,
				\frac{(2j+1)!! (2k+1)!!}{(2i+1)!!} \\
			&\quad
			+
			\delta_{a,b} \, \delta_{j-i+1 \geq 0} \,
				\sum_{\substack{m,m'\geq 0 \\ m+m'= j-i+1}}
					\frac{M_{\theta} M_{\phi}}{\rho^{m+m'+1 + j}}
					\frac{(2j+1)!! (2k-1)!!}{(2i+1)!!} \\
			&\quad
			+
			\delta_{a,c} \, \delta_{k-i+1 \geq 0} \,
				\sum_{\substack{m,m' \geq 0 \\ m+m' = k-i+1}}
					\frac{M_{\theta} M_{\phi}}{\rho^{m+m'+1 + k}}
					\frac{(2j-1)!! (2k+1)!!}{(2i+1)!!} \\
			&\quad
			+
			\delta_{i,0} \,
			\frac{M_{\theta} M_{\phi}^2}{\rho^{j+k+1}}
			(2j-1)!! (2k-1)!!\,.
	\end{split}
	\end{equation}
	The sum of the first and last lines is bounded by $(2M_{\theta} + 2\rho^2M_{\theta}M_{\phi}^2) \cdot \PIL{C}^{j,k}_i\big(\tfrac{\rho}{3}\big)$, while the sum of the second and third lines is bounded by
	\begin{equation}\label{eq:bound:Mth:Mphi}
			\delta_{j+k-i+2 \geq 0} \frac{M_{\theta} M_{\phi}}{\rho^{j+k-i+2}} \frac{(2j+1)!! (2k+1)!!}{(2i+1)!!}
			\bigl(
				(j-i+2) + (k-i+2) 
			\bigr)\,.
	\end{equation}
	If $j + k - i + 2 \geq 0$, we have $(j - i + 2) + (k - i + 2) \leq 2(j + k - i + 2) \leq 2 \cdot 3^{(j + k - i + 2)/3}$. It follows that \eqref{eq:bound:Mth:Mphi} is bounded by
	\begin{equation}
		\delta_{j + k - i + 2 \geq 0} \frac{2M_{\theta}M_{\phi}}{(3^{-1/3}\rho)^{j + k - i + 2}} \, \frac{(2j + 1)!!(2k + 1)!!}{(2i + 1)!!}
		\leq
		4\rho M_{\theta}M_{\phi} \cdot \PIL{C}_i^{j,k}(3^{-4/3}\rho)\,.
	\end{equation}
	Since $u \mapsto \PIL{C}_{i}^{j,k}(u)$ are decreasing functions of $u$, we obtain $\big|C_{\alpha}^{\beta,\gamma}\big| \leq 2M_{\theta}(1 + \rho M_{\phi})^2 \cdot \PIL{C}_i^{j,k}(3^{-4/3}\rho)$. Finally, for the D-case we find 
	\begin{equation}
			\big|D_{\alpha}\big|
			\leq
			2M_{\theta} (1 + 4 M_{\phi} \rho)
			\cdot
			\PIL{D}_{i}(\tfrac{\rho}{3})\,.
	\end{equation}
	We conclude by choosing $Q = 2M_{\theta} \max{ \left\{(1 + \rho M_{\phi})^2,1 + 4\rho M_{\phi}\right\}}$, $u = 3^{-4/3}\rho$, and using the fact that the coefficients of $\PIL{A}, \PIL{B}, \PIL{D}$ are weakly decreasing functions of $u$.
\end{proof}

\begin{corollary}[{Upper bound by PI amplitudes}]\label{prop:upper:bound:Fgn}
	Consider a bounded spectral curve. The corresponding amplitudes $F_{g;\alpha_1,\dots,\alpha_n}$ satisfy for $2g - 2 + n > 0$ and $\alpha_i = (a_i,k_i) \in \mf{A}$
	\begin{equation}
		\big|F_{g;\alpha_1,\dots,\alpha_n}\big|
		\le
		|\mf{a}|^{3g-3+n} \,
		Q^{2g-2+n} \,
		\PIR{F}_{g;k_1,\dots,k_n}(u)\,.
	\end{equation}
	with constants $Q, u > 0$ from \cref{lem:upper:bound:ABCD}.
\end{corollary}

\begin{proof}[{Proof of \cref{prop:upper:bound:Fgn}}]
	We proceed by induction on $2g-2+n$. For the base cases, \cref{lem:upper:bound:ABCD} yields
	\begin{equation}
		\big|F_{0;\alpha,\beta,\gamma}\big|
		=
		\big|A_{\alpha,\beta,\gamma}\big|
		\le
		Q \cdot
		\PIL{A}_{i,j,k}(u)\,,
		\qquad
		\big|F_{1;\alpha}\big|
		=
		\big|D_{\alpha}\big|
		\leq
		Q \cdot
		\PIL{D}_{i}(u)
		\leq
		|\mf{a}| \cdot
		Q \cdot
		\PIL{D}_{i}(u)\,.
	\end{equation}
	For the general case, we start from the recursion equation \labelcref{eq:TR:QAS}:
	\begin{equation}
	\begin{split}
		\big|F_{g;\alpha_1,\dots,\alpha_n}\big|
		& \le
		\sum_{m=2}^n \sum_{\mu \in \mf{A}}
			\big|B_{\alpha_1,\alpha_m}^{\mu}\big| \cdot \big|F_{g;\mu,\alpha_2,\dots,\widehat{\alpha_m},\dots,\alpha_n}\big| \\
		& \qquad
		+
		\frac{1}{2} \sum_{\mu,\nu \in \mf{A}}
			\big|C_{\alpha_1}^{\mu,\nu}\big| \Biggl(
				\big|F_{g-1;\mu,\nu,\alpha_2\dots,\alpha_n}\big|
				+
				\!\!\!
				\sum_{\substack{g_1 + g_2 = g \\ A_1 \sqcup A_2 = \{\alpha_2,\dots,\alpha_n\}}}
				\!\!\!
				\big|F_{g_1;\mu,A_1}\big| \cdot \big|F_{g_2;\nu,A_2}\big|
			\Biggr) \\
		& \le
		\sum_{m=2}^n \sum_{\substack{\mu \in \mf{A} \\ \mu = (a,i)}}
			Q
			\cdot
			\PIL{B}_{k_1,k_m}^{i}(u)
			\cdot
			|\mf{a}|^{3g-3+n-1}
			\cdot
			Q^{2g-2+n-1}
			\cdot
			\PIR{F}_{g;i,k_2,\dots,\widehat{k_m},\dots,k_n}(u) \\
		& \qquad
		+
		\frac{1}{2} \sum_{\substack{\mu, \nu \in \mf{A} \\ \mu = (a,i) \\ \nu = (b,j)}}
			Q \cdot \PIL{C}_{k_1}^{i,j}(u)
			\cdot
			|\mf{a}|^{3g-3+n-2}
			\cdot
			Q^{2g-2+n-1} \\
			&
			\qquad\qquad
			\times
			\Biggl(
				\PIR{F}_{g-1;i,j,k_2\dots,k_n}(u)
				+
				\!\!\!
				\sum_{\substack{g_1 + g_2 = g \\ K_1 \sqcup K_2 = \{k_2,\dots,k_n\}}}
				\!\!\!
				\PIR{F}_{g_1;i,K_1}(u)
				\cdot
				\PIR{F}_{g_2;j,K_2}(u)
			\Biggr) \\
		& =
		|\mf{a}|^{3g-3+n} \,
		Q^{2g-2+n} \,
		\PIR{F}_{g;k_1,\dots,k_n}(u)\,.
	\end{split}
	\end{equation}
	The first inequality is simply the triangular inequality, the second inequality is a combination of \cref{lem:upper:bound:ABCD} and the induction hypothesis, the last equality is the recursive definition of the PI amplitudes $\PIR{F}_{g;k_1,\dots,k_n}(u)$.
\end{proof}

\subsection{Lower bounds}
\label{subsec:lower:bound}
In this section we provide a lower bound for the amplitudes of spectral curves satisfying a positivity condition.

\begin{definition}\label{def:pos}
	A regular spectral curve $\mc{S}$ is \emph{positive} if the expansion coefficients \eqref{eq:expns:coeff:y:B} are non-negative, i.e. for any $\alpha,\beta \in \mf{A}$:
	\begin{equation}
		\theta_{\alpha} \geq 0
		\qquad\text{and}\qquad
		\phi_{\alpha,\beta} \geq 0\,.
	\end{equation}
	It is \emph{strongly positive} if there exists furthermore some $M_-,\rho_- > 0$ such that for any $(a,i) \in \mf{A}$:
	\begin{equation}
		\theta_{(a,i)} \geq \frac{M_-}{\rho_-^{i + 1}}\,.
	\end{equation}
\end{definition}

Note that being regular then imposes $\theta_{(a,0)} > 0$ for any $a \in \mf{a}$. 

In the positive case the amplitudes playing a fundamental role are those associated with the \emph{Airy spectral curve} $\AiR{\mc{S}}$ on $\Sigma = \P^1$ given by
\begin{equation}\label{eq:SC:Airy}
	x(z) = z^2,
	\qquad
	y(z) = -\frac{z}{2}\,,
	\qquad
	\omega_{0,2}(z_1,z_2) = \frac{\dd z_1 \dd z_2}{(z_1 - z_2)^2}\,.
\end{equation}
Again, since there is only one ramification point, we omit its dependence from the associated quantum Airy structure. A simple application of the formulae from \cref{subsec:TR:QAS} yields the following expressions.

\begin{lemma}\label{lemma:QAS:Airy} 
	The quantum Airy structure $(\AiL{A},\AiL{B},\AiL{C},\AiL{D})$ associated with $\AiR{\mc{S}}$ is given by
	\begin{equation}\label{eq:QAS:Airy}
		\begin{split}
			\AiL{A}_{i,j,k}
			&= \delta_{i,j,k,0} \,,\\
			\AiL{B}_{i,j}^{k}
			&= \delta_{i+j,k+1} \frac{(2k+1)!! }{(2i+1)!! (2j-1)!!}\, ,\\
			\AiL{C}_{i}^{j,k}
			&= \delta_{i,j+k+2} \frac{(2j+1)!! (2k+1)!!}{(2i+1)!!} \,,\\
			\AiL{D}_{i}
			&= \frac{\delta_{i,1}}{24}\, .
		\end{split}
	\end{equation}
\end{lemma}

Denote by $\AiR{F}_{g;k_1,\dots,k_n}$ the amplitudes associated with $\AiR{\mc{S}}$; it is not hard to see that they vanish unless $k_1+\cdots+k_n = 3g-3+n$. The above lemma yields a lower bound on the amplitudes associated with any positive spectral curve.

\begin{proposition}[Lower bound by Airy amplitudes]\label{prop:lower:bound:Fgn:pos}
	Consider a positive regular spectral curve. The corresponding amplitudes $F_{g;\alpha_1,\dots,\alpha_n}$ satisfy for $2g - 2 + n > 0$ and $\alpha_i = (a_i,k_i) \in \mf{A}$
	\begin{equation}
		F_{g;\alpha_1,\dots,\alpha_n}
		\geq
		\delta_{a_1,\dots,a_n} \,
		Q_-^{2g - 2 + n} \,
		\AiR{F}_{g;k_1,\dots,k_n}\,,
	\end{equation}
	where we can take $Q_- = \min\set{ \theta_{(a,0)} | a \in \mf{a} }$.
\end{proposition}

\begin{proof}
	In view of the recursion \eqref{eq:TR:QAS}, the amplitudes $F_{g;\alpha_1,\dots,\alpha_n}$ are weakly increasing functions of the coefficients of the quantum Airy structure $(A,B,C,D)$, in the range where the latter are non-negative. The positivity assumption on the spectral curve $\mc{S}$ implies that the coefficients \eqref{eq:ABCD:explicit} of $(A,B,C,D)$ are non-negative, and in fact, each term in \eqref{eq:ABCD:explicit} is non-negative. Therefore, the amplitudes $F_{g;\alpha_1,\dots,\alpha_n}$ of $\mc{S}$ are lower bounded by the amplitudes obtained by using the recursion \eqref{eq:TR:QAS} with initial data on the right-hand side of the following inequalities:
	\begin{equation}
	\begin{split}
		A_{\alpha,\beta,\gamma}
		& \geq
		Q_- \, 
		\delta_{a,b,c}\delta_{i,j,k,0}\,, \\
		B_{\alpha,\beta}^{\gamma}
		& \geq
		Q_- \, 
		\delta_{a,b,c} \delta_{i+j,k + 1} \frac{(2k + 1)!!}{(2i + 1)!!(2j - 1)!!}\,, \\
		C_{\alpha}^{\beta,\gamma}
		& \geq
		Q_- \, 
		\delta_{a,b,c} \delta_{i,j+k+ 2} \frac{(2k + 1)!!(2j + 1)!!}{(2i + 1)!!}\,, \\
		D_{\alpha}
		& \geq
		Q_- \, 
		\frac{\delta_{i,1}}{24}\,,
	\end{split}
	\end{equation}
	with $\alpha = (a,i),\beta = (b,j),\gamma = (c,k)$ and $Q_- = \min\set{ \theta_{(a,0)} | a \in \mf{a} }$. These initial data form a quantum Airy structure, namely $q$ times the one associated with the disjoint union of $|\mf{a}|$ spectral curves $\AiR{\mc{S}}$. The amplitudes of the latter are $\delta_{a_1,\dots,a_n} \cdot Q_-^{2g - 2 + n} \cdot \AiR{F}_{g;k_1,\dots,k_n}$. The thesis then follows by induction on $2g-2+n$. 
\end{proof}

In the strongly positive case, we can use a better lower bound by the PI amplitudes.

\begin{proposition}[Lower bound by PI amplitudes]\label{prop:lower:bound:Fgn:strng:pos}
	Consider a strongly positive regular spectral curve. The corresponding amplitudes $F_{g;\alpha_1,\dots,\alpha_n}$ satisfy for $2g - 2 + n > 0$ and $\alpha_i = (a_i,k_i) \in \mf{A}$
	\begin{equation}\label{eq:lower:bound:Fgn:strng:pos}
		F_{g;\alpha_1,\dots,\alpha_n}
		\geq
		\delta_{a_1,\dots,a_n} \,
		Q_-^{2g - 2 + n} \,
		\PIR{F}_{g;k_1,\dots,k_n}(u_-) \,,
	\end{equation}
	where we can take $Q_- = 2M_-$ and $u_- = \rho_-/3$.
\end{proposition}

\begin{proof}
	As before, positivity implies that the coefficients of the quantum Airy structure are increasing functions of the coefficients $(\theta_{\alpha})$ and $(\phi_{\alpha,\beta})$. Thus, $F_{g;\alpha_1,\dots,\alpha_n}$ is lower bounded by the amplitudes of the spectral curve where the coefficients $(\phi_{\alpha,\beta})$ are set to zero. This implies that all terms that do not obey $a_1 = \cdots = a_n$ can be ignored. For the PI spectral curve \eqref{eq:SC:PI}, we have the coefficients $\theta_i(u) = \frac{1}{2(3u)^{i + 1}}$. Hence, the strong positivity assumption implies $\theta_{(a,i)} \geq Q_- \cdot \theta_{i}(u_-)$ with $Q_- = 2M_-$ and $u_- = \rho_-/3$. This implies the comparison \eqref{eq:lower:bound:Fgn:strng:pos} again by monotonicity of amplitudes with respect to non-negative coefficients of quantum Airy structures.
\end{proof}

\begin{remark}
	The lower bound by Airy or Painlev\'e I are in general non-optimal. When there is more than one ramification point, an annoying feature of these lower bounds is that they are trivial unless $a_1 = \cdots = a_n$. It would be interesting to find non-trivial lower bounds without this feature.
\end{remark}

\section{Intersection theory}
\label{sec:TR:CohFT}

\subsection{From topological recursion to intersection theory}
\label{subsec:Mgn}
The amplitudes of a (local) regular spectral curve can always be expressed in terms of intersection indices of tautological classes on $\Mbar_{g,n}$, see \cite{Eyn14,DOSS14,DNOPS18}. We shall make use of this theory in the more restrictive case of spectral curves $\mc{S} = (\Sigma, x, y, \omega_{0,2})$ such that $\dd x$ and $\omega_{0,2}$ extend respectively as meromorphic $1$-form and fundamental bidifferential on a compact Riemann surface $\overline{\Sigma}$ containing $\Sigma$ ($y$ need not extend to $\overline{\Sigma}$). In this case, the tautological class involved is a cohomological field theory obtained by Givental--Teleman reconstruction from a topological field theory, an $R$-matrix, and a translation. We refer to \cite{Pan19} for definitions and notations regarding cohomological field theories. 

 To make the statement precise, recall the centred local coordinate $\zeta_a$ on $U_a$ such that $\zeta_a(a) = 0$ and $x(z) = x(a) + \zeta_a(z)^2$, and the coefficient $t_{(a,1)} = -2 \frac{\dd y}{\dd \zeta_a}(a)$. Consider the auxiliary (multivalued) functions $\varphi^a \colon \Sigma \to \C$ and the associated (single-valued) meromorphic differentials $\widehat{\xi}^{(a,i)}$ defined as
\begin{equation}\label{eq:basis:xi:hat}
	\varphi^a(z)
	\coloneqq
	\int^z_a \left.\frac{\omega_{0,2}(w,\cdot)}{\dd\zeta_{a}(w)}\right|_{w = a}\,,
	\qquad
	\widehat{\xi}^{(a,i)}(z)
	\coloneqq
	\dd\biggl( \left( -2 \frac{\dd}{\dd x(z)} \right)^{i} \varphi^{a}(z) \biggr)\,.
\end{equation}
Define a unital, semi-simple topological field theory on the vector space $V \coloneqq \bigoplus_{a \in \mf{a}} \C.\mr{e}_a$ by setting
\begin{equation}
	\eta(\mr{e}_a,\mr{e}_b) \coloneqq \delta_{a,b}\,,
	\qquad
	\mb{1}
	\coloneqq
	\sum_{a \in \mf{a}} t_{(a,1)} \mr{e}_a\,,
	\qquad
	w_{g,n}(\mr{e}_{a_1} \otimes \cdots \otimes \mr{e}_{a_n})
	\coloneqq
	\frac{\delta_{a_1,\ldots,a_n}}{(t_{(a_1,1)})^{2g-2+n}}
\end{equation}
as the pairing, the unit and the topological field theory respectively. Define the $R$-matrix $R \in \id_V + \lambda.\End(V)\bbraket{\lambda}$ and the translation $T \in \lambda^2.V\bbraket{\lambda}$ by the formulae
\begin{equation}
\begin{split}
	R^{-1}(\lambda)^b_a
	& \coloneqq
	- \sqrt{\frac{\lambda}{2\pi}} \int_{\gamma_b}
		e^{-\frac{x(z) - x(b)}{2\lambda}} \,
		\widehat{\xi}^{(a,0)}(z) \,, \\
	T(\lambda)^a
	& \coloneqq
	t_{(a,1)} \lambda
	+
	\frac{1}{\sqrt{2\pi \lambda}} \int_{\gamma_a}
		e^{-\frac{x(z) - x(a)}{2\lambda}} \,
		\omega_{0,1}(z)
	\,.
\end{split} 
\end{equation}
Here $\gamma_c$ is the formal steepest descent path for $x(z)$ emanating from the ramification point $c$; locally it can be taken along the real axis in the $\zeta_c$-plane. Moreover, the equations are intended as equalities between formal power series in $\lambda$, where on the right-hand side we take an asymptotic expansion as $\lambda \to 0$ along the positive real axis. Through the Givental action, we can then define a cohomological field theory 
\begin{equation}
	\Omega_{g,n}
	\coloneqq
	\hat{R}\hat{T}w_{g,n}
	\in
	H^{\textup{even}}(\Mbar_{g,n}) \otimes (V^{\ast})^{\otimes n}
\end{equation}
from the data $(w,R,T)$ via a sum over stable graphs. In general this cohomological field theory has no unit. The link with the topological recursion correlators will be given through the intersection indices
\begin{equation}\label{eq:int:indices}
	\bigl\langle
		\tau_{k_1}(\mr{e}_{a_1}) \cdots \tau_{k_n}(\mr{e}_{a_n})
	\bigr\rangle^{\!\Omega}_{\!g}
	\coloneqq
	\int_{\Mbar_{g,n}} \Omega_{g,n}(\mr{e}_{a_1} \otimes \cdots \otimes \mr{e}_{a_n}) \prod_{i=1}^n \psi_i^{k_i}\,.
\end{equation}

\begin{theorem}[{\cite{Eyn14,DOSS14,DNOPS18}}] \label{thm:corresp}
	Let $\mc{S} = (\Sigma, x, y, \omega_{0,2})$ be a regular spectral curve such that $\dd x$ and $\omega_{0,2}$ extend to a compact Riemann surface respectively as meromorphic form and fundamental bidifferential. The correlators are given by
	\begin{equation}
		\omega_{g,n}(z_1,\dots,z_n)
		=
		\sum_{\substack{
			k_1,\dots,k_n \geq 0 \\
			a_1,\dots,a_n \in \mf{a} \\
			k_1+\cdots+k_n \leq 3g-3+n
		}}
		\bigl\langle
			\tau_{k_1}(\mr{e}_{a_1}) \cdots \tau_{k_n}(\mr{e}_{a_n})
		\bigr\rangle^{\!\Omega}_{\!g} \,
		\prod_{i=1}^n \widehat{\xi}^{(a_i,k_i)}(z_i) \,.
	\end{equation}
\end{theorem}

\begin{remark}
	The basis $(\widehat{\xi}^{\alpha})_{\alpha \in \mf{A}}$ defined in \eqref{eq:basis:xi:hat} is in general different from the basis $(\xi^{\alpha})_{\alpha \in \mf{A}}$ defined in \eqref{eq:basis:xi}. The assumptions of \cref{thm:corresp} guarantee the existence of a change of basis relating $(\xi^{\alpha})_{\alpha}$ to $(\widehat{\xi}^{\alpha})_{\alpha}$, and one could also perform this change of basis at the level of the quantum Airy structure.
\end{remark}

The simplest example is the computation of the cohomological field theory associated with the Airy spectral curve. In this case the associated 1-dimensional cohomological field theory (whose underlying vector space is generated by $\mb{1} = \mr{e}_0$) is the Poincaré dual of the fundamental class of $\Mbar_{g,n}$. Besides, the differential forms $\xi^{\alpha}$ and $\widehat{\xi}^{\alpha}$ coincide. Thus, the intersection indices \eqref{eq:int:indices} (evaluated on $\mb{1}^{\otimes n}$) are simply intersections of $\psi$-classes and coincide with the amplitudes of the corresponding quantum Airy structure:
\begin{equation} 
	\AiR{F}_{g;k_1,\dots,k_n}
	=
	\int_{\Mbar_{g,n}} \psi_1^{k_1} \cdots \psi_n^{k_n}
	\eqqcolon
	\braket{\tau_{k_1} \cdots \tau_{k_n}}_g\, .
\end{equation}
This is a well-known restatement of Witten's conjecture/Kontsevich's theorem \cite{Wit91,DVV91,Kon92}.

In \cite{DOSS14}, a version of \cref{thm:corresp} directly concerning Frobenius manifolds was considered. We recall from \cref{sec:local:SP} that series $R(\lambda),T_R(\lambda)$ can be associated with semi-simple point of a Frobenius manifold $M$. The tangent space of $M$ at this point is a Frobenius algebra and we take $w_{g,n}$ to be the corresponding topological field theory. We can then form the unital cohomological field theory via the Givental action
\begin{equation}\label{eq:Omega:Givental}
	\Omega_{g,n} \coloneqq \hat{R}\hat{T}_R w_{g,n} \,.
\end{equation}
Then the ancestor potential of the Frobenius manifold is a generating series for the intersection indices defined by \eqref{eq:int:indices} with $\mr{e}_a = \del_{v_a}$, where $(\del_{v_a})_{a \in \mf{a}}$ is the basis of normalised canonical vector fields. If we replace $T_R(\lambda)$ in \eqref{eq:Omega:Givental} with an arbitrary series $T(\lambda) \in \mr{O}(\lambda^2)$ we still obtain a cohomological field theory albeit it may not have a unit anymore. In any case, the construction of the local spectral curve in \cref{sec:local:SP} is tailored to express the following. 

\begin{theorem}[\cite{DOSS14}] \label{thm:corresp:Frob}
	The topological recursion correlators of the local spectral curve associated with a semi-simple point of a Frobenius manifold satisfy
	\begin{equation}
		\omega_{g,n}(z_1,\dots,z_n)
		=
		\sum_{\substack{
			k_1,\dots,k_n \geq 0 \\
			a_1,\dots,a_n \in \mf{a} \\
			k_1+\cdots+k_n \leq 3g-3+n
		}}
		\bigl\langle
			\tau_{k_1}(\del_{v_{a_1}}) \cdots \tau_{k_n}(\del_{v_{a_n}})
		\bigr\rangle^{\!\Omega}_{\!g} \,
		\prod_{i=1}^n \widehat{\xi}^{(a_i,k_i)}(z_i) \,.
	\end{equation}
\end{theorem}

\subsection{The Painlev{\'e} I case}
A less trivial example of application of \cref{thm:corresp} is that of the PI spectral curve $\PIR{\mc{S}}(u)$ introduced in \eqref{eq:SC:PI}. The expression of the free energies of $\PIR{\mc{S}}(u)$ in terms of intersection numbers was first observed in \cite{IZ92} (its mathematical status is explained in \cite[Section~2.2]{GM07}) and the expression of the general $n$-point correlators in terms of intersection numbers can be found in \cite{CED18}. We revisit it here for completeness. Since the vector space underlying the topological field theory is 1-dimensional (again, generated by $\mb{1} = \mr{e}_0$), we can write without confusion $\Omega_{g,n} = \Omega_{g,n}(\mb{1}^{\otimes n})$. We also omit the ramification-point subscript from the differential forms: $\widehat{\xi}^i = \widehat{\xi}^{(0,i)}$.

Recall the definition of the \emph{multi-index $\kappa$-class}: for $\mu = (\mu_1,\dots,\mu_m)$ a $m$-tuple of non-negative integers
\begin{equation}
	\kappa_{\mu}
	\coloneqq
	p_{m,\ast}\bigl( \psi_{n+1}^{\mu_1+1} \cdots \psi_{n+m}^{\mu_m+1} \bigr)\,,
\end{equation}
where $p_m \colon \Mbar_{g,n+m} \to \Mbar_{g,n}$ is the morphism forgetting the last $m$ marked points and stabilising.

\begin{lemma}\label{prop:int:numbs:PI} 
	The cohomological field theory and the differential forms $(\widehat{\xi}^{i})_{i \geq 0}$ associated with $\mc{S}^{\textup{PI}}(u)$ are given by
	\begin{equation}\label{eq:CohFT:PI}
		\Omega_{g,n}
		=
		\frac{1}{(6u)^{2g - 2 + n}}
		\sum_{m = 0}^{3g-3+n} \frac{\kappa_{\bm{1}^m}}{m!} \, u^{-m}\,,
		\qquad\qquad
		\widehat{\xi}^i(z) = \frac{(2i+1)!!}{z^{2i+2}} \dd z\,.
	\end{equation}
	Besides, the amplitudes of the corresponding quantum Airy structure are given by
	\begin{equation}\label{eq:int:numbs:PI}
		\PIR{F}_{g;k_1,\dots,k_n}(u)
		=
		\frac{\braket{\tau_{k_1} \cdots \tau_{k_n} \tau_2^{3g-3+n-|k|}}_{g}}{6^{2g-2+n} \, (3g-3+n-|k|)!} \,
		u^{-(5g-5+2n-|k|)}
	\end{equation}
	for $|k|\coloneqq k_1 + \cdots + k_n \leq 3g-3+n$ (including the $n=0$ case), and zero otherwise.
\end{lemma}

\begin{proof}
	The local coordinate $\zeta$ is actually the global coordinate $z$. The topological field theory is determined by the single constant $t_1 = 6u$, so it reads $w_{g,n} = (6u)^{-(2g-2+n)}$. Moreover, a simple computation shows that $\varphi(z) = - 1/z$, so $\widehat{\xi}^i$ is given by \eqref{eq:CohFT:PI} and equal to $\xi^i$.

	The $R$-matrix is the identity while the translation is simply a multiple of $\lambda^2$:
	\begin{equation}
	\begin{split}
		R^{-1}(\lambda)
		&=
		- \frac{1}{2 \sqrt{2\pi \lambda}} \int_{\R}
 			e^{-\frac{z^2}{2\lambda}} \dd z
 		= 1 \,,\\
		T(\lambda)
		&=
		6u \lambda + \frac{1}{\sqrt{2\pi \lambda}} \int_{\R}
 			e^{-\frac{z^2}{2\lambda}} \,
 			2z(z^3 - 3uz) \dd z
 		= 6 \lambda^2\,.
 	\end{split}
	\end{equation}
	As a consequence, the cohomological field theory (for $2g-2+n > 0$ and $n > 0$) is given by
	\begin{equation}
		\Omega_{g,n} 
		=
		\sum_{m \geq 0} \frac{1}{m!} p_{m,\ast} \bigl(
			w_{g,n + m} \cdot 6\psi_{n+1}^2 \cdots 6\psi_{n+m}^2
		\bigr)
		=
		 \frac{1}{(6u)^{2g - 2 + n}} \sum_{m = 0}^{3g-3+n} \frac{\kappa_{\bm{1}^m}}{m!} \, u^{-m}\,.
	\end{equation}
	We arrive to
	\begin{equation}
	\begin{split}
		\PIR{F}_{g;k_1,\dots,k_n}(u)
		&=
		\int_{\Mbar_{g,n}} \Omega_{g,n} \prod_{i=1}^n \psi_i^{k_i} \\
		&=
		\frac{1}{(6u)^{2g - 2 +n}} 
		\sum_{m = 0}^{3g-3+n} \frac{u^{-m}}{m!} \int_{\Mbar_{g,n+m}}
			\prod_{i=1}^n \psi_i^{k_i} \prod_{j=1}^{m} \psi_{n+j}^2 \,.
	\end{split}
	\end{equation}
	The second equality follows from the projection formula. Since the complex dimension of $\Mbar_{g,n+m}$ is $3g - 3 + n + m$ while each $\psi$-class has complex cohomological degree $1$, the only term contributing to the sum is $m = 3g-3+n-|k|$.

	As for the free energies, that is $n=0$, one can directly compute them from \cref{def:Fg} (or equivalently from \cref{lem:Fg} with $t_1 = 6u$ and $t_2 = -2$):
	\begin{equation}
		\PIR{F}_g
		=
		\frac{1}{2g - 2}\left( 6u \, \PIR{F}_{g;1}(u) - 6 \, \PIR{F}_{g;2}(u) \right)
		=
		\frac{3}{g-1}
		\frac{
			\braket{\tau_1 \tau_2^{3g - 3}}_g - (3g-3) \braket{\tau_2^{3g - 3}}_g
		}{6^{2g-1} (3g-3)!} u^{-(5g-5)} \,,
	\end{equation}
	where we inserted \eqref{eq:int:numbs:PI} to get the last line. The claim then follows from the dilaton equation: $\braket{\tau_1 \tau_2^{3g - 3}}_g = (5g-5) \braket{\tau_2^{3g - 3}}_g$.
\end{proof}

\subsection{Upper bound on Painlev{\'e} I amplitudes}
\label{sec:upPI}
We now establish an upper bound on the PI amplitudes by exploiting their intersection-theoretic expression and the known results on the asymptotic behaviour of $\psi$-class intersections.

\begin{proposition}\label{lemma:bound:PI} 
	There exists $S(u),P(u) > 0$ such that for any $2g - 2 + n > 0$ and $k_1,\ldots,k_n \geq 0$ we have
	\begin{equation}\label{eq:bound:PI}
		\PIR{F}_{g;k_1,\dots,k_n}(u)
		\prod_{i=1}^n (2k_i + 1)!!
		\leq
		S(u) \,
		\left( \frac{2u}{27} \right)^g \,
		P(u)^{3g - 3 + n} \,
		\frac{(3g - 3 + n)!}{g!} \,.
	\end{equation}
	One can take $P(u)$ to be the function
	\begin{equation}
		P(u)
		= 
		\begin{cases}
			\frac{2}{5u^2} & \text{if } u < \tfrac{2}{5} \,, \\
			\frac{5}{10u - 2} & \text{if } u \geq \tfrac{2}{5} \,.
		\end{cases}
	\end{equation}
\end{proposition}

As the proof will show, for $|k|$ bounded independently of $g$ and $n$ one can get a better upper bound. \cref{lemma:bound:PI} is interesting for large $|k|$, and a variant of its proof will give the following useful bound.
 
\begin{proposition}\label{lem:bound:sums:PI}
	For any $v > 0$, $2g - 2 + n > 0$, and $k_1,\ldots,k_n \geq 0$, we have
	\begin{equation}
		\sum_{k_1,\dots,k_n \geq 0}
			v^{-|k|} \,
			\PIR{F}_{g;k_1,\dots,k_n}(u)
			\prod_{i=1}^n (2k_i + 1)!!
		\leq
		\frac{S(\tfrac{u}{v})}{v} \,
		\left( \frac{2u}{27} \right)^g \,
		\left( \frac{P(\frac{u}{v})}{v^2} \right)^{3g - 3 + n} \,
		\frac{(3g - 3 + 2n)!}{g! \, n!} \,.
	\end{equation}
\end{proposition}
 
\begin{proof}[Proof of \cref{lemma:bound:PI}]
	Aggarwal established a uniform upper bound for intersection indices of $\psi$-classes, see \cite[Proposition~1.2]{Agg21}:
	\begin{equation}\label{eq:Agg:upper}
		\braket{ \tau_{d_1} \cdots \tau_{d_n} }_g \,
		\prod_{i=1}^n (2d_i + 1)!!
		\le
		\frac{(2(3g-3+n)+1)!!}{24^g \, g!}
		\left( \frac{3}{2} \right)^{n-1}\,.
	\end{equation}
	We use it in \eqref{eq:int:numbs:PI} to get
	\begin{equation}\label{eq:bound:master}
	\begin{split}
		&\PIR{F}_{g;k_1,\dots,k_n}(u)
		\prod_{i=1}^n (2k_i+1)!!
		=
		\frac{
			\braket{\tau_{k_1} \cdots \tau_{k_n} \tau_2^{3g-3+n-|k|}}_g
		}{
			6^{2g-2+n} (3g-3+n-|k|)!
		}
		u^{-(5g-5+2n-|k|)}
		\prod_{i=1}^n (2k_i+1)!! \\
		& \quad \leq
		\frac{15^{-(3g-3+n-|k|)}}{6^{2g-2+n} (3g-3+n-|k|)!}
		\frac{(2(6g-6+2n-|k|)+1)!!}{24^g \, g!}
		\left( \frac{3}{2} \right)^{3g-4+2n-|k|}
		u^{-(5g-5+2n-|k|)} \\
		& \quad =
		\frac{3}{8u}
		\frac{(\frac{2}{27} u)^g}{(40 \, u^2)^{D_{g,n}}}
		\frac{c_{D_{g,n},|k|}(u)}{g!} \,,
	\end{split}
	\end{equation}
	where we have set
	\begin{equation}
		c_{D,k}(u) \coloneqq (10u)^k \frac{(4D+1 - 2k)!!}{(D - k)!}\,.
	\end{equation}
	The main reason behind the rewriting of the above bound as in the last line of \eqref{eq:bound:master} is to have a bound as a function of $(g,D_{g,n})$ rather than $(g,n)$. This will make a further analysis easier. We now look for the maximum of the expression $c_{D,k}(u)$, seen as a function of $k \in \set{0,1,\ldots,D}$. To this end, consider the ratio
	\begin{equation}
		r_{D,k}(u)
		\coloneqq
		\frac{c_{D,k}(u)}{c_{D,k-1}(u)}
		=
		10u \frac{D - k}{4D + 1 - 2k} \,,
	\end{equation}
	which is a decreasing function of $k$ in the real segment $[0,D - 1]$ with maximum at $r_{D,0}(u) = 10u \frac{D}{4D+1} \leq \frac{5}{2}u$. We consider two separate cases.
	\begin{itemize}
		\item
		If $u \leq \frac{2}{5}$, then $r_{D,k}(u) \le 1$. We deduce that $k \mapsto c_{D,k}(u)$ is a decreasing sequence of, therefore
		\begin{equation}
			c_{D,k}(u) \leq c_{D,0}(u) = \frac{(4D+1)!!}{D!} \le 16^{D} \cdot D! \,.
		\end{equation}
		The last inequality follow from the fact that $\frac{(4D+1)!!}{16^D (D!)^2}$ is a decreasing function of $D \ge 0$, hence it is bounded by its value at $D = 0$.

		\item
		If $u > \frac{2}{5}$, then $k \mapsto r_{D,k}(u)$ equals one at $k = D \cdot \frac{5u-2}{5u-1} - \frac{1}{10u - 2} = D \cdot \lambda - \epsilon$, where $\lambda = \lambda(u) \coloneqq \frac{5u-2}{5u-1} \in (0,1)$ and $\epsilon = \epsilon(u) \coloneqq \frac{1}{10u - 2} \in (0,\frac{1}{2})$. We deduce that
		\begin{equation}
			c_{D,k}(u)
			\leq
			c_{D,D \lambda - \epsilon}(u)
			=
			(10u)^{D \lambda - \epsilon}
			\frac{
				(2D(2-\lambda) + 1 + 2\epsilon)!!
			}{
				(D(1-\lambda) + \epsilon)!
			} \,,
		\end{equation}
		where the right-hand side is intended after extension to the reals via the $\Gamma$-function. Similarly as before, this can be bounded by
		\begin{equation}\label{eq:bound:cDk}
			c_{D,k}(u)
			\le
			s(u) \cdot p(u)^{D} \cdot D! 
		\end{equation}
		for some irrelevant $s(u) > 0$ independent of $D$ and 
		\begin{equation}
			\forall u \in (\tfrac{2}{5},+\infty)
			\qquad
			p(u)
			\coloneqq
			(10u)^{\lambda(u)} \frac{(4 - 2\lambda(u))^{2 - \lambda(u)}}{(1 - \lambda(u))^{1 - \lambda(u)}}
			=
			\frac{(10 u)^2}{5u-1}\,.
		\end{equation}
	\end{itemize}
	Note that $p(\frac{2}{5}) = 16$. Therefore, we can define $p(u) = 16$ for $u < \frac{2}{5}$, making the inequality \eqref{eq:bound:cDk} valid for any $u > 0$ for an appropriate choice of $s(u) > 0$. Together with \eqref{eq:bound:master}, it implies the desired bound with $P(u) \coloneqq \frac{p(u)}{40u^2}$ and $S(u) \coloneqq \frac{3}{8u} s(u)$.
\end{proof}

\begin{proof}[Proof of \cref{lem:bound:sums:PI}]
	Considering again \eqref{eq:bound:master}, the proof of \cref{lemma:bound:PI} can be easily adapted to give the bound
	\begin{equation}
		v^{-|k|} \,
		\PIR{F}_{g;k_1,\dots,k_n}(u)
		\prod_{i = 1}^n (2k_i + 1)!!
		\leq
		S(\tfrac{u}{v})
		\left( \frac{2u}{27} \right)^g
		\left( \frac{P(\frac{u}{v})}{v^2} \right)^{3g - 3 + n}
		\frac{(3g - 3 + n)!}{g!} \,,
	\end{equation}
	where the exponential growth rate appears through $\frac{1}{v^2} \, P(\frac{u}{v}) = \frac{1}{40u^2} \, p(\frac{u}{v})$. Moreover, we know that the coefficients $\PIR{F}_{g;k_1,\dots,k_n}(u)$ vanish for $|k| > D_{g,n}$ for cohomology degree reasons, and there are $\binom{D_{g,n} + n}{D_{g,n}}$ tuples $(k_1,\dots,k_n) \in \Z_{\geq 0}$ such that $|k| \leq D_{g,n}$. Hence 
	\begin{equation}
		\sum_{k_1,\dots,k_n \geq 0}
			v^{-|k|} \,
			\PIR{F}_{g;k_1,\dots,k_n}(u)
			\prod_{i = 1}^n (2k_i + 1)!!
		\leq
		S(\tfrac{u}{v})
		\left( \frac{2u}{27} \right)^g
		\left( \frac{P(\frac{u}{v})}{v^2} \right)^{3g - 3 + n}
		\frac{(3g - 3 + 2n)!}{g! \, n!} \,.
	\end{equation}
\end{proof}

\subsection{Large genus asymptotic of Painlev{\'e} I free energies}
\label{sec:PIfree}

We conclude this section with a comment on the free energies. An asymptotic equivalent for the free energies of PI has been established by Kapaev in \cite{Kap04} using Riemann--Hilbert methods (see \cite{BSSV23} for the full transseries structure):
\begin{equation}\label{eq:Kapaev}
	\PIR{F}_{g}
	\sim
	\frac{1}{\sqrt{30g}}
	\biggl( \frac{5}{24\sqrt{3}} \biggr)^{2g - 2}
	\frac{\Gamma(2g-2)}{(2\pi)^{3/2} \, u^{5g - 5}}
\end{equation}
as $g \rightarrow \infty$. More precisely, Kapaev computes the asymptotics of the formal (also known as $0$-parameter) solution of the PI equation:
\begin{equation}
	\hbar^2 \del_s^2 U = 6 U^2 + s \,,
	\qquad\qquad
	U(s;\hbar)
	=
	\sum_{g \geq 0} U_g \, \hbar^{2g} \Bigl( -\frac{s}{6} \Bigr)^{\frac{1-5g}{2}} \,.
\end{equation}
It was then shown in \cite{IS16} (see \cite{IMS18} for the analogous result for all Painlev{\'e} equations) that the coefficients $U_g$, determined from the initial condition $U_0 = -1$, are related to the free energies as
\begin{equation}
	\PIR{F}_{g}
	=
	\frac{144 \cdot 2^{2g-2}}{(5g-3)(5g-5)u^{5g - 5}} \, U_g \,.
\end{equation}
This identity relating the intersection number in the left-hand side and the coefficients of the PI solution is in fact known since \cite{IZ92}. The asymptotic equivalence \eqref{eq:Kapaev} is remarkably accurate, see \cref{fig:Fg:PI}, and gives access to the large genus asymptotics of $\braket{\tau_2^{3g - 3}}_g$.

\begin{proposition}\label{prop:asym:tau:2}
	As $g \to \infty$:
	\begin{equation}\label{eq:asym:tau:2}
		\braket{\tau_2^{3g - 3}}_g
		\sim
		\frac{1}{4\pi} \biggl( \frac{3}{10 \sqrt{5}} \biggr)^{2g - 2} \Gamma(5g - 5) \,.
	\end{equation}
\end{proposition}

\begin{proof}
	This follows from the specialisation of \eqref{eq:int:numbs:PI} to $n = 0$, in conjunction with Kapaev's asymptotic \eqref{eq:Kapaev} and Stirling's formula.
\end{proof}

It is worth noticing that these are intersections of $\psi$-classes in $\Mbar_{g,3g - 3}$, hence fall in the regime $n = \mr{O}(g)$ which is not covered by \cite{Agg21}.

We also remark that Kapaev's asymptotic is not sufficient to obtain (uniform) upper bounds for the amplitudes $\PIR{F}_{g;k_1,\dots,k_n}$ in the shape of \cref{lemma:bound:PI} for $n > 0$, and we do not know if the Riemann--Hilbert method can be pushed further to provide them. In comparison, the starting point of the proof of \cref{lemma:bound:PI} was the uniform upper bound \cite[Proposition 1.2]{Agg21} for $\psi$-class intersections, which would give here
\begin{equation}
	\braket{\tau_2^{3g - 3}}_g
	\leq
	\frac{(12g - 11)!!}{15^{3g - 3} \, 24^g \, g!} \left(\frac{3}{2}\right)^{3g - 4}
	\sim
	\frac{1}{3\pi} \sqrt{\frac{5}{2}}
	\biggl( \frac{72\sqrt{3}}{625} \biggr)^{2g - 2}
	\Gamma(5g - 5) \,.
\end{equation}
This is weaker than the more precise asymptotics from \cref{prop:asym:tau:2}, namely the exponential growth rate of the upper bound is larger by a factor of about $1.5$ (see again \cref{fig:Fg:PI}):
\begin{equation}
	0.1342 \simeq \frac{3}{10 \sqrt{5}}
	<
	\frac{72\sqrt{3}}{625} \simeq 0.1995 \,.
\end{equation}
One could expect that having uniform asymptotics of $\PIR{F}_{g;k_1,\dots,k_n}(u)$ could lead to an improvement of the upper bound of \cref{lemma:bound:PI}. 

\begin{figure}
\centering
	\begin{tikzpicture}
	\begin{axis}[
		xlabel = {$g$},
		axis y line = left, 
		axis x line = bottom,
		xmin = 0,
		xmax = 60,
		x=.9*\textwidth/60,
		yticklabel style = {
			/pgf/number format/sci
		},
		ymode=log,
		yminorticks=true,
	]
	\addplot[thick, color=ForestGreen, samples=100, smooth, domain = 2:60] {(2*pi*sqrt(60*pi*x))^(-1)*(24*sqrt(3)/5)^(2-2*x)*sqrt(2*pi)*(2*x-2)^((2*x-2)-0.5)*exp(-(2*x-2))*exp(1/(12*(2*x-2)))};
	\addplot[thick, color=BrickRed, samples=100, smooth, domain = 2:60] {(375*sqrt(2))/(pi)*13500^(-x)*(6*x - 4.5)^(6*x - 5)*(3*x - 2)^(5/2 - 3*x)*(x + 1)^(0.5 - x)*exp(3.5 - 2*x)*exp((1/12)*(1/(6*x - 4.5) - 1/(3*x - 2) - 1/(x + 1)))};
	\addplot[NavyBlue, mark = *, mark size=1.5, only marks] table[col sep=&,row sep=\\,y expr={\thisrow{y}}] {
		x & y \\
		2 & 7/51840 \\
		3 & 245/26873856 \\
		4 & 259553/116095057920 \\
		5 & 1337455/1114512556032 \\
		6 & 245229441961/216661240892620800 \\
		7 & 1551927036905/935976560656121856 \\
		8 & 1670346976424371/479219999055934390272 \\
		9 & 6952919330793301075/698702758623552341016576 \\
		10 & 2684145230737949205740569/72441502014089906716598599680 \\
		11 & 947682426779747304018947/5433112651056743003744894976 \\
		12 & 200853500428633643281495272377959/198283402120870216349151290311311360 \\
		13 & 28852511263185177273720650811245/4055796861563254425323549120004096 \\
		14 & 577561061192136191896260223014016860617/9718338207803407723783275459389014671360 \\
		15 & 27554979246667911396283341465592563685/47306814593273799616973876935727775744 \\
		16 & 664879209267280088258320326407460951970252207777/100449700598387252107015199963059318561544601600 \\
		17 & 602178075005491929915968055179745913266003145/6975673652665781396320499997434674900107264 \\
		18 & 3628033794193483714980314311177712581409817845924189587/2832777988587094961179851373550237320381376847937536 \\
		19 & 2544247007844006623043659150118618828913855769229070325/118629087412683376993342810852373794034812943597568 \\
		20 & 1585842837104238847065743794160714657989249006479200099402985847/3939067775489993015948761989024058942705187209615843000320 \\
		21 & 62132046751952795548978891760598012295954434281106677548954379/7379678269768207516001869577504468979317643595317510144 \\
		22 & 45285469252658164484620398770189447292000039638771344549721034948339611/232144605642014038928975291823039781933709323518109605502648320 \\
		23 & 1397874544138988959990598794292272612418875753631024367913249793242555/280545849103508176928327073858409892717739538919590465634304 \\
		24 & 1113389877826091729058739602162704223827917040740354162349117207209526956246858543/7971570898533681952198919614457667632501768657772838411188028173189120 \\
		25 & 1466344231092694177966473925191826186430706885389756026337305208217415364753305/342697670463882174617926550171466902085090046110521352383147016192 \\
		26 & 153858244369778319928269153542773888563089186785904113245736039094299312567145266065542629/1078033657513007145396500850817788234962326140571158112789532296951627776000 \\
		27 & 201303432388384651735306104152327670879081245711228788619430954166012480498007609719465/38972949875834536426249078991699901972724780403872930278421011269419008 \\
		28 & 316293844253814138772959537406982888515187484510488725671669627690796135538610297280453092949879585/1564429748858525096227257845512755056303072778844233287322233059803873326162509824 \\
		29 & 6366066058936504574259862042702508144636637905447350024856638954487321622887893152586655065075/745978235654127643693570063358666923667465581342808383618465928937851584512 \\
		30 & 728848650930897248653263106462560016137335420993807790969402670543202166738785522406456826101586371314671661/1881520885795161830931368283704044353155029994242613243863001834229400802895736615206912 \\
		31 & 1969393537533302732270536202757379534834258722554012729548990952946311478460795115647866683946701121449/104413081688036938032501614628185891971887822489390203838309439141573210580975616 \\
		32 & 1311621304714296993964410032087678687191995568491305932494272037103829354332561663052348160741933542882633022396967439/1334589494867092686909769307208274730943673115342410675825698634010489938222594017148520103936 \\
		33 & 272756885263608689059540480795384029537805275660981134965714816475857018022521608206107698889299981881540619255355/4988412660978740613984660820062864359260567803970549094686889157211381465974782300258304 \\
		34 & 5156449637639915691762324712696518971297790929054596221214003516877700935771305438000160760307145224815419597020728457558234049/1590811390910809951908582576235823824217778909437970512681562774547341532184160335469433301407629312 \\
		35 & 2495416436072125614542075379119450756343511425347281863243746656299071704001224274505291080917648273050375692839677563405/12211634194075957023034449687513891951469869984119904183793504656853461828706267071032328192 \\
		36 & 7381611037674253990617897821643886788678881672578039855737277013595524691816039019130684996639353432313614216036555621257723029946768283/539788336395161811754522390377153427946943220354457834470334493158405574731521342630268218985268023132160 \\
		37 & 21060260074305446696105282096257016584826535939929999296552260366900565305090101098268837925319673749332083198712306093917282744765/21717377130194558080095208496128235093604724034026803867070423536415324386583382265295448769036288 \\
		38 & 36155809141729633870299018497726571154027468932443959035462771198321329455064245577612826842998907704954584265642343242960645502061184776192659245/496971461890959344587254867136613014596706969006789671907612408625891791094897499298687135426206187108480057344 \\
		39 & 73995587211080952573923974905789691253862274544203262815098406298292455408839667822284430374483293690775781426258408651812566933319597198775/12834448666893859155942345936207879862678145390733024442145810619856582744685500917615243131729989337088 \\
		40 & 125631527115307710275561173174223696593294381153121140221449815317229069126567225548448208583159041400122341672759346718232192338731508046011160417958107009/260693714021880898900599801128333507906863986592648757727810022900795585962556532285088960757436528652568382461181952 \\
		41 & 1188618265142244911198937837085219950919249838760404034868282725215573919854558009927647338048860431618830886563156955541413602854298423207989298197/28014223953338006679749524771916483877104633981288420508666897790878536820399847055544177008374006199353344 \\
		42 & 1205906113185622104993008313915793798443847036734915421882387598346472919269192867197579609033628552667636722317166660865656929619843956964530124768385830475855085947/306879977093907620912838004417959787040620610767465500346151927118055552741803409465856943185881419023620583714747760771072 \\
		43 & 2248886074484226890124804589336342705839725450571693532159409727938688207719960219630774881374753602795844477393604606445584599383133121219789179742685221285/5881642347451221178426772224913409622965872113639466462635632524990530562516588689005611051262139349566633279488 \\
		44 & 347404832913379853062916794559956357169043994361484950208842795659159228170041588436931123370295216468512681776331799757170542549766231439010681798382869478114450059524093789/8898501392384668044607810763033008396591716629254929380573843041006829186703204111980891436676578715758479552132923947837554688 \\
		45 & 67266784640257733100489206808150450552116276966572312498776316950180997518613267367294040840101800641419715312300210184141254562256547778878140767178823980412218128015/16098949114610804950973196448966188976320282787592492907136046837042880710253405763289086300184267244945009414217138176 \\
		46 & 43941400971297989525896243512423088854814477056692314031134238702597181143455519865063856143367410915469562139993694404868205709535639566498367145858651693471476574882882811844674174617/93847386859565650759331116843996775497289396136044530603491100171028049559034790627725632019821894832251081062500678167074488576901120 \\
		47 & 14364983416378548509844900970233798153211419188479987174734631269051339746114052475860859983145028787749254633544550173073251016957806127783790076729568850362552995280694053815/261751350113628476715604930773154996912448175999576561041406154973767739024309192177520591510486893474902276701264248045568 \\
		48 & 5901794075097313625376199264345953588941691835513057923424452818323401035218299817605892347641385391018248895029760895426502133279639114377468300000427809880239535462207829152479845350102286881435/878043959560682944366445105541355252084360964726591147231208943659672661208052812807581239950037697232191157710700744987572415722936476893184 \\
		49 & 350449907611931976679367555943758694506832383893258702763407361065393085026671921285943486898183663527570908554491830659766321665985005501968491573426264355144952289118054419538744094875/407783657435317989608957673187694843635105906624653713762432408479925680926658882901739989176095493630832660761851038836224688128 \\
		50 & 16443253407597068574429077835892087470448833901310318719806519296394401735428065414564559122292507093809905999808858810561555793796534508498921947040983234122936767107027875136923333162383157951339124233/143476119349490319419878689586332092255487068278728085334801801857580554001656289411962211123750840739217618962088972797969279845790471117864960 \\
		51 & 757537641355380902648534379950311467544320151229665524109034672385903992098934954069985663049120339990122288084231957730687469526786886911335132537843264509876270107364053171843876281450046177533/47563885803255490307988823000612726561598752948699609173250116125098531423285492101658952337499778377100321551262305169857247623249920 \\
		52 & 37764982645014132157617533171845176481652936986807575387258848563135209613570180214956361886591745077487823746111719296429932656461715325270944420678370834016666534138037326793554982480758907701464231490154006882017/16386965426057320199306249626000099915673867401835002320883495499559897345723123253393437841520432664094838339762706185182309955747770139723063340564480 \\
		53 & 333625544570416182811653540814725995184390937373014648023756928188594536880970454920374320793253864152262592264497462972812601456013928425005383979233365573061254859385933749279073761117944278921603954465/961627617615898200850794827897187860531778947615629218421101547770792068903416765114500034778635519316863460994800781002105889547817582592 \\
		54 & 3106147714635814646314965917267644945390832950016496375600958224325947131703225571643268462267641426335534312146985789562984408394658892851956922129106121704103867196077350137971589548039856104066490349074130568559513029889331/57206348784933045959432634741437668331690876820001037673891443858055145750290460293243084006354066688840186181277666985582433258335767632276781021911926702080 \\
		55 & 14806536205877024971272508537577009962578031745132810536567764355623639445169209555993625512610268105138545270744761621214542596077610829188045142086299804435328356627016760376325229217758763422046296801764360790235/1677286868458373106078678167409261665768596137681387019378637547155224641280330532329731324661473933198024975013868051221465127539468222064492544 \\
		56 & 44016473971233298044604073414628660267241558188447756607781811208018449221493897902976375623860668455737450493389300281809049819140961972434644477074329308940087313984248638312982811799049674381659994975020844060659409952813215782618557/29543862639505105021500557556187848443287612616909335900220247799804407226546987541994508278143404970847602506878590329230363841117030763438048250507221508764467200 \\
		57 & 72578490777679691208604512513308768667614308137935726491253259998932789078120450466363177644445966618360725881063819836860894538412145199494959372411541221766120064383586474011209033973970403834586400805090929574819883640485/278241764034822597821179763835187599211020876087689929870683425422485905292711471347114469985442587830353967055100543281026407077267304421834539139072 \\
		58 & 108212865998392001445424027954443277665940210510260959155678983106430576193159104761400515568138664479685209332533496088496539684137826125314225405857207275299572719300484830022025680300950763284898868354364940193533171918493179242060895400494377/2285635256802873025551662858486269517012117256509316439806793461434689225978799892168141363165915520937668285846483435427857331088171714212102924014920751203960842354688 \\
		59 & 93101236342292512901918609454801084497525421822755660990404851397673927495872799823466370110375809175591897600771052709439058707052094332276448167776383898504838511828002640943794416030281822510770662428207768210400579191547320651475/10457438459484772516511220243981690728747008606879738324259765861078710264521267939109950239932874221016023497798898818674096483592014369390229319002882048 \\
		60 & 45919339012900304571361113296472587046034245611070765012623608361959633839771091328365075389516600923473559966924875480295964015438304656568745895359401405927367434509877713895723475296142248437268999912949033914331019064435862199217386479384354711947877831/26491272900809666479697534977711949661293992507209353267881766688274167457299984751034995380851656076409276804727403531098966680908356338520052915147656469566938217234185584640 \\
	};
	\end{axis}
	\end{tikzpicture}
	\caption{
		A log-linear plot of the Painlev{\'e} I free energies $\PIR{F}_{g}$ (for $u = 1$) in blue, and Aggarwal's bound \eqref{eq:bound:master} in red, and Kapaev's asymptotic \eqref{eq:Kapaev} in green.
	}
	\label{fig:Fg:PI}
\end{figure}
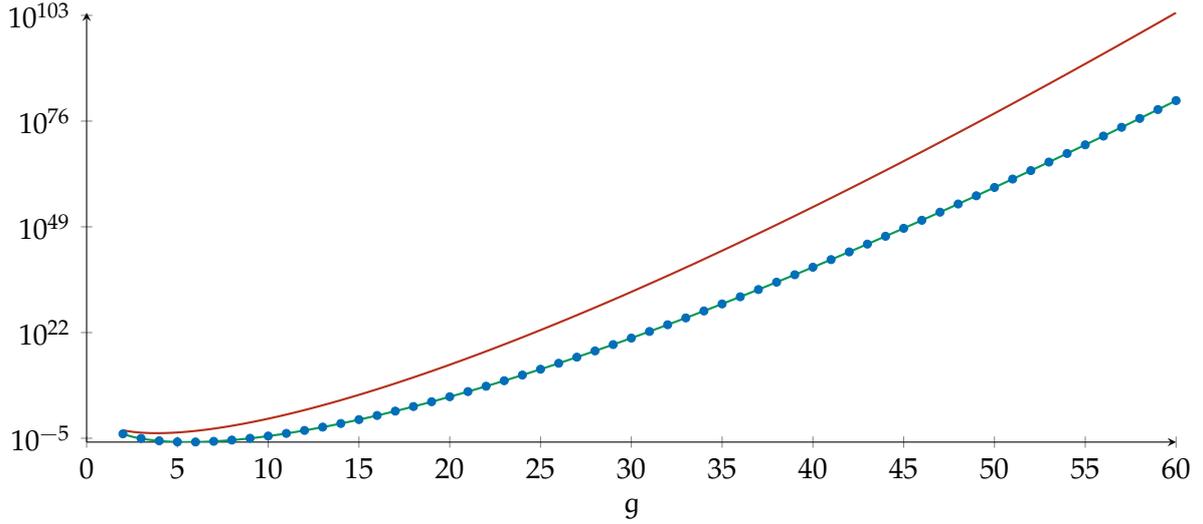

\section{Proof of the main theorems}
\label{sec:proof}


\subsection{Amplitudes}
We start with a uniform upper bound on the amplitudes for $n > 0$. The free energies will be discussed shortly.

\begin{theorem}\label{thm:final:bound:ampl}
	Consider a bounded spectral curve. The corresponding amplitudes $F_{g;\alpha_1,\dots,\alpha_n}$ satisfy for $2g - 2 + n > 0$, $n > 0$, and $\alpha_i = (a_i,k_i) \in \mf{A}$
	\begin{equation}
		\big|F_{g;\alpha_1,\dots,\alpha_n}\big|
		\prod_{i=1}^n (2k_i+1)!!
		\leq
		S(u)
		\left( \frac{2u}{27} \right)^g
		Q^{2g - 2 + n} \,
		\bigl( |\mf{a}| \, P(u) \bigr)^{3g - 3 + n}
		\frac{(3g - 3 + n)!}{g!} \,,
	\end{equation}
	where $Q$ comes from \cref{prop:upper:bound:Fgn} and $P(u)$ from \cref{lemma:bound:PI}. In particular, if $n$ is fixed, there exist constants $\ms{S}_n(u), \ms{A}(u) > 0$ such that uniformly in $\alpha_1,\dots,\alpha_n$
	\begin{equation}
		\big|F_{g;\alpha_1,\dots,\alpha_n}\big|
		\prod_{i=1}^n (2k_i+1)!!
		\leq
		\ms{S}_n(u) \,
		\frac{\Gamma(2g - 2 + n)}{\ms{A}(u)^{2g-2+n}} \,.
	\end{equation}
	One can take $\ms{A}(u)^{-1} = Q\sqrt{u/2} \, \bigl( |\mf{a}| \, P(u)\bigr)^{3/2}$.
\end{theorem}

\begin{proof}
	The first claim is a direct consequence of the comparison with the PI amplitudes (\cref{prop:upper:bound:Fgn}) and the upper bound for the latter (\cref{lemma:bound:PI}). For fixed $n$, we use the Stirling inequalities to find
	\begin{equation}
		\frac{(3g - 3 + n)!}{g!}
		\leq
		M_n \biggl( \frac{3\sqrt{3}}{2} \biggr)^{2g - 2 + n}
		\Gamma(2g-2+n)
	\end{equation}
	for some constant $M_n > 0$, and rewrite $3g - 3 + n = \frac{3}{2}(2g - 2 + n) - \frac{n}{2}$ and $g = \frac{1}{2}(2g - 2 + n) + 1 - \frac{n}{2}$ to get the second claim. 
\end{proof}

\subsection{Generalised periods}
In the previous subsection, we provided an upper bound for the amplitudes associated with a (bounded local) spectral curve. However, interesting enumerative information are often stored in generalised periods of the correlators (rather than in the amplitudes). In this section we derive similar upper bounds bearing on generalised periods on spectral curves. To this end, we will need again a boundedness assumption on the generalised period under consideration. Throughout this section, we will only consider spectral curves and not their local version.

\begin{definition}\label{def:bounded:periods}
	We say that a linear form $I$ on $H^0(\Sigma,K_{\Sigma}(*\mf{a}))$ is bounded if there exists $N_I,v_I > 0$ such that for any $(a,i) \in \mf{A}$:
	\begin{equation}
		\big|I[\xi^{(a,i)}]\big| \leq N_I \frac{(2i + 1)!!}{v_I^{i}} \,.
	\end{equation}
\end{definition}

\begin{lemma}\label{lem:bound:periods}
	The following linear forms on $H^0(\Sigma,K_{\Sigma}(*\mf{a}))$ are bounded:
	\begin{itemize}
		\item
		$\mc{F} \colon \omega \mapsto \sum_{a \in \mf{a}} \Res_{z = a} \big(\int_{a}^{z} y \dd x\big)\omega(z)$;

		\item
		$I_{(p,k)}$ for each $p \in \Sigma \setminus \mf{a}$ and $k \in \Z_{\geq 0}$ --- see \eqref{eq:gen:period};

		\item
		$\mr{ev}_{(Z,p)}$ for each $p \in \Sigma \setminus \mf{a}$ --- see \eqref{eq:gen:period:ev};

		\item
		$I_{(Z,p,k)}$ for each $p \in \Sigma \setminus \mf{a}$ --- see \eqref{eq:gen:period2};

		\item
		$\int_{\gamma}$ for $\gamma$ a loop or a path between two points in $\Sigma \setminus \mf{a}$.
	\end{itemize}
	In the third and fourth cases, the constants in \cref{def:bounded:periods} can be chosen uniformly for $p$ in any compact in the domain of definition of the local coordinate $Z$. In the fifth case, the constants can be chosen uniformly for loops or paths remaining in a compact subset of the universal cover of $\Sigma \setminus \mf{a}$.
\end{lemma}

\begin{proof}
	We only discuss the first two cases, as the others are obtained similarly. For the linear form $\mc{F}$, we have
	\begin{equation}
		\mc{F}[\xi^{(a,i)}]
		=
		\sum_{b \in \mf{a}}
			\Res_{z = b} \left(\int_{b}^{z} y\dd x\right) \xi^{(a,i)}(z)
		=
		\Res_{z = a} \left(\int_{a}^z y \dd x\right) \xi^{(a,i)}(z)
		=
		-(2i - 1)!! \, t_{(a,i)}
	\end{equation}
	using the series expansion of $y \dd x$ from \eqref{eq:ydx:expns} and that of $\xi^{(a,i)}$ from \eqref{eq:xi:expansion}. As \cref{lem:reg:SC:bnd} guarantees that regular spectral curves are bounded, we deduce that $\mc{F}$ is bounded with constant $v_{\mc{F}} = \rho_t$.

	As for the linear form $I_{(p,k)}$, recall the definition of the local coordinate $X_p$ near $p$ in \eqref{eq:coord:X}. We have
	\begin{equation}
		I_{(p,k)}[\xi^{(a,i)}]
		=
		(2i + 1)!! \,
		\Res_{z = p} \
		\Res_{w = a} \
			\frac{ \omega_{0,2}(z,w)}{X_p^{k}(z) \zeta_a(w)^{2i + 2}} \,.
	\end{equation}
	We can then move the contours in $z$ and in $w$, but avoiding intersections due to the double pole of $\omega_{0,2}$ along the diagonal. Hence, we get the desired bound for some choice of $N_{I_{(p,k)}} > 0$ and $v_{I_{(p,k)}} = \min_{a \in \mf{a}} v_{p,a}$. Recalling the notations from the proof of \cref{lem:reg:SC:bnd}, if $p \notin U_{R_a}$ we can take any $v_{p,a} < R_{a,+}^2$ (for instance $v_{p,a} = \rho$), and if $p \in U_{R_{a}}$ we should rather take $v_{p,a} < |\zeta_a(p)|^2$ and choose the contour in $z$ so that its $\zeta_a$-projection remains in $\mb{D}(R_a) \setminus \overline{\mb{D}(v_{p,a})}$.
\end{proof}

\begin{theorem}\label{thm:bound:periods}
	Consider a regular spectral curve and let $I_1,\ldots,I_n$ be bounded linear forms on the space $H^0(\Sigma,K_{\Sigma}(*\mf{a}))$. 
	Then, for any $2g - 2 + n > 0$, we have
	\begin{equation}
		\big|I_{1} \otimes \cdots \otimes I_n[\omega_{g,n}]\big|
		\leq
		\left( \prod_{i = 1}^{n} |\mf{a}| \, N_{I_i} \right)
		\frac{S(\tfrac{u}{v_I})}{v_I}
		\left( \frac{2u}{27} \right)^g
		Q^{2g-2+n}
		\left( \frac{|\mf{a}| \, P(\frac{u}{v_I})}{v_I^2} \right)^{3g - 3 + n} \,
		\frac{(3g - 3 + 2n)!}{g! \, n!} ,
	\end{equation}
	where $v_I \coloneqq \min_i v_{I_i}$. In particular, for fixed $n$, there exist $\ms{S}_{n,I}, \ms{A}_{I} > 0$ depending only on $n$ and the constants in the boundedness assumptions, such that
	\begin{equation}\label{eq:bound:periods}
		\big|I_{1} \otimes \cdots \otimes I_n[\omega_{g,n}]\big|
		\leq
		\ms{S}_{n,I}(u) \,
		\frac{\Gamma(2g - 2 + n)}{\ms{A}_{I}^{2g-2+n}} \,.
	\end{equation}
	One can take any fixed $\ms{A}_{I}^{-1} = Q/v_I^{3} \, \sqrt{u/2} \, \bigl( |\mf{a}| \, P(u/v_I) \bigr)^{3/2}$.
\end{theorem}

\begin{proof}
	We have
	\begin{equation}
		(I_{1} \otimes \cdots \otimes I_n)[\omega_{g,n}]
		=
		\sum_{\alpha_1,\dots,\alpha_n \in \mf{A}} F_{g;\alpha_1,\dots,\alpha_n} \prod_{i = 1}^n I_i[\xi^{\alpha_i}] \,.
	\end{equation}
	Coming back to the comparison between the amplitudes with the PI amplitudes (\cref{prop:upper:bound:Fgn}) and using \cref{lem:bound:sums:PI}, we get
	\begin{equation}
	\begin{split}
		& \quad \big|I_{1} \otimes \cdots \otimes I_n[\omega_{g,n}]\big| \\
		& \leq
		\sum_{\substack{(a_1,k_1),\dots,(a_n,k_n) \in \mf{A} \\ k_1 + \cdots + k_n \leq 3g - 3 + n}}
			\big|F_{g;(a_1,k_1),\dots,(a_n,k_n)}\big|
			\prod_{i = 1}^n
				\frac{N_{I_i} (2k_i + 1)!!}{v_{I_i}^{k_i}} \\
		& \leq
		\left( \prod_{i = 1}^n |\mf{a}| \, N_{I_i} \right)
		|\mf{a}|^{3g - 3 + n} \,
		Q^{2g - 2 + n}
		\sum_{k_1,\dots,k_n \geq 0 }
			v_I^{-|k|} \,
			\PIR{F}_{g;k_1,\dots,k_n}(u) \\
		& \leq
		\left( \prod_{i = 1}^{n} |\mf{a}| \, N_{I_i} \right)
		\frac{S(\tfrac{u}{v_I})}{v_I}
		\left( \frac{2u}{27}\right)^g
		Q^{2g-2+n}
		\left( \frac{|\mf{a}| \, P(\frac{u}{v_I})}{v_I^2} \right)^{3g - 3 + n} \,
		\frac{(3g - 3 + 2n)!}{g! \, n!} .
	\end{split} 
	\end{equation}
	Specialising to $n$ fixed follows again from the Stirling-type inequality
	\begin{equation}
		\frac{(3g - 3 + 2n)!}{g! \, n!}
		\leq
		M_n \biggl(\frac{3\sqrt{3}}{2}\biggr)^{2g - 2 + n} \Gamma(2g-2+n) \,.
	\end{equation}
\end{proof}


\begin{corollary}\label{cor:cor}
	Consider a regular spectral curve and let $Z_1,\dots,Z_n$ be local coordinates in some open of $\Sigma \setminus \mf{a}$. Take $p_i$ in a compact subset in the domain of definition of the local coordinate $Z_i$. There exists $\ms{S}_{n,Z}, \, \ms{A}_{Z} > 0$ depending only on these compacts and the constants in the boundedness property such that for any $2g-2+n > 0$
	\begin{equation}
		\bigg|
			\frac{\omega_{g,n}(z_1,\dots,z_n)}{\dd Z_1(z_1) \cdots \dd Z_n(z_n)}\bigg|_{z_i = p_i}
		\leq
		\ms{S}_{n,Z} \,
		\frac{\Gamma(2g - 2 + n)}{\ms{A}_{Z}^{2g-2+n}} \,.
	\end{equation}
\end{corollary}

\subsection{Free energies and wave function}
We conclude this section with an analysis of the free energies and the stable coefficients of the wave function.

\begin{theorem}\label{thm:free}
	Consider a (bounded local) regular spectral curve. Then, the free energies satisfy
	\begin{equation}
		\forall g \geq 2 \,,
		\qquad\qquad
		|F_{g}|
		\leq
		\ms{S}_{0} \,
		\frac{\Gamma(2g - 2)}{\ms{A}_0^{2g-2}}
	\end{equation}
	for some constants $\ms{S}_{0}, \, \ms{A}_{0} > 0$ depending only on the constants in the boundedness assumption. One can take $\ms{A}_{0}^{-1} = Q/\rho_t^{3} \, \sqrt{u/2} \, \bigl(|\mf{a}| \, P(u/\rho_t) \bigr)^{3/2}$.
\end{theorem}

\begin{proof}
	For regular spectral curves, this is a simple consequence of \cref{thm:bound:periods} due to the boundedness of linear form $\mc{F}$ in \cref{lem:bound:periods}. For bounded local spectral curves, we can just redo the proof starting from the expression of the free energies in terms of the coefficients $(t_{\alpha})$, \cref{lem:Fg}, in combination with the geometric bound on the coefficients $t_{\alpha}$ in the boundedness assumption and the comparison with PI amplitudes. We have
	\begin{equation}
	\begin{split}
		|F_{g}|
		& \leq
		\frac{1}{2g - 2}
		\sum_{(a,i) \in \mf{A}}
			(2i - 1)!! \frac{M_t}{\rho_t^{i + 1}}
			\big|F_{g;(a,i)}\big| \\
		& \leq
		\frac{M_t \, |\mf{a}|^{3g - 2} \, Q^{2g - 1}}{(2g - 2) \, \rho_t}
		\sum_{(a,i) \in \mf{A}}
			(2i + 1)!! \, \rho_t^{-i} \, \PIR{F}_{g;i}(u) \,.
	\end{split}
	\end{equation}
	The result is then implied by the Stirling-type inequality
	\begin{equation}
		\frac{(3g - 1)!}{g!}
		\leq
		M_0 \biggl(\frac{3\sqrt{3}}{2}\biggr)^{2g - 2} \Gamma(2g-2)
	\end{equation}
	and by \cref{lem:bound:sums:PI}.
\end{proof}

To conclude, we provide an upper bound for the stable coefficients of the wave function (see \cref{def:wave}).

\begin{theorem}\label{thm:wave}
	Consider a regular spectral curve and take a path remaining in a given compact subset of the universal cover of $\Sigma \setminus \mf{a}$ between two points $z_0$ and $z$. There exists $\ms{S}_{\psi}, \, \ms{A}_{\psi} > 0$ depending only on this compact and the constants in the boundedness property such that
	\begin{equation}
		\forall \chi \geq 1 \,,
		\qquad\qquad
		\big|f_{z_0,\chi}(z)\big|
		\leq
		\ms{S}_{\psi} \, \frac{\Gamma(\chi)}{\ms{A}_{\psi}^{\chi}} \,.
	\end{equation}
	One can take $\ms{A}_{\psi}^{-1} = Q/v^{3} \, \sqrt{2u} \, \bigl( |\mf{a}| \, P(u/v) \bigr)^{3/2}$.
\end{theorem}

\begin{proof}
	Recall the definition
	\begin{equation}
		f_{z_0,\chi}(z) = \sum_{\substack{g \geq 0, \ n > 0 \\ 2g - 2 + n = \chi}} \frac{1}{n!} \int_{z_0}^z \cdots \int_{z_0}^z \omega_{g,n}\,.
	\end{equation}
	By \cref{lem:bound:periods}, the generalised period $\int_{z_0}^{z}$ is bounded. Consider the associated constants $N$ and $v$, which can be taken to be uniform for paths remaining in a given compact subset of the universal cover of $\Sigma \setminus \mf{a}$ between two points $z_0$ and $z$. By \cref{thm:bound:periods}, we have an upper bound for the individual terms:
	\begin{equation}
	\begin{split}
		\big|f_{z_0,\chi}(z)\big|
		& \leq
		\frac{S(\tfrac{u}{v})}{v}
		\sum_{\substack{g \geq 0, \ n > 0 \\ 2g - 2 + n = \chi}}
			\frac{(|\mf{a}| \, N)^n}{n!}
			\left( \frac{2u}{27} \right)^g
			Q^{2g-2+n}
			\left( \frac{|\mf{a}| \, P(\frac{u}{v})}{v^2} \right)^{3g - 3 + n} \,
			\frac{(3g - 3 + 2n)!}{g! \, n!} \\
		& \leq
		\frac{S(\tfrac{u}{v})}{v} \, Q^{\chi} \, \Gamma(\chi)
		\sum_{\substack{g \geq 0, \ n > 0 \\ 2g - 2 + n = \chi}}
			\frac{(|\mf{a}| \, N)^n}{n!}
			\left( \frac{2u}{27} \right)^g
			\left( \frac{|\mf{a}| \, P(\frac{u}{v})}{v^2} \right)^{3g - 3 + n} \,
			3^{3g - 3 + 2n} \\
		& =
		\ms{S} \, \frac{\Gamma(\chi)}{\ms{A}_{\psi}^{\chi}}
		\sum_{\substack{g \geq 0, \ n > 0 \\ 2g - 2 + n = \chi}}
			\frac{\ms{T}^n}{n!}
	\end{split}
	\end{equation}
	for some constants $\ms{S}, \, \ms{A}_{\psi}, \, \ms{T} > 0$. One can take $\ms{A}_{\psi}$ as in the statement of the theorem. To conclude, the sum over $(g,n)$ can be bounded by
	\begin{equation}
		\sum_{\substack{g \geq 0, \ n > 0 \\ 2g - 2 + n = \chi}}
			\frac{\ms{T}^n}{n!}
		\leq
		\sum_{n \ge 0} \frac{\ms{T}^n}{n!}
		=
		e^{\ms{T}} \,.
	\end{equation}
	This would simply adjust the overall constant.
\end{proof}

\section{Examples of applications}
\label{sec:exmpls}
Unless explicitly stated, in all subsequent examples we take $\Sigma = \C$ and $\omega_{0,2}(z_1,z_2) = \frac{\dd z_1 \dd z_2}{(z_1 - z_2)^2}$.

\textbf{Weil--Petersson volumes.}
Consider the spectral curve $x(z) = z^2$ and $y(z) = - \frac{\sin(2\pi z)}{4\pi}$. There is a single ramification point at $z = 0$, so we omit the dependence on $\alpha \in \mf{a}$. In this case, the associated amplitudes compute the Weil--Petersson intersection numbers \cite{Mir07b,EO}:
\begin{equation}
	F_{g;k_1,\ldots,k_n} = \int_{\Mbar_{g,n}} e^{2\pi^2 \kappa_1} \prod_{i=1}^n \psi_i^{k_i}
	=
	\braket{e^{2\pi^2 \kappa_1} \tau_{k_1} \cdots \tau_{k_n}}_g \,.
\end{equation}
The constants in the boundedness condition can be taken to be $M_{\theta} = \frac{1}{2}$, $M_{\phi} = 0$, and $\rho = \frac{1}{4}$. Thus, 
\begin{equation}
	\braket{e^{2\pi^2 \kappa_1} \tau_{k_1} \cdots \tau_{k_n}}_g
	\prod_{i=1}^n (2k_i+1)!!
	\le
	\ms{S} \, \frac{\Gamma(2g-2+n)}{\ms{A}^{2g-2+n}}
\end{equation}
for some $\ms{S} > 0$ and $\ms{A}^{-1} = 1728 \cdot 3^{1/3} \cdot 5^{-3/2} \simeq 222.910$. The same inequality holds for $n = 0$, as we can take $\rho_t = 1$. It is worth comparing our estimate with the one obtained by Grushevski \cite{Gru01} using ad hoc methods for Weil--Petersson volumes and valid for $k_1 = \cdots = k_n = 0$ only. The latter is slightly worst, with an exponential factor of $\ms{A}_{\textup{Gru}}^{-1} = 768 \cdot e^{-1} \sqrt{6} \simeq 692.058$. It is also worth comparing it with the optimal one obtained by Mirzakhani--Zograf \cite{MZ15} (again, only valid for $k_1 = \cdots = k_n = 0$) that is $\ms{A}_{\textup{MZ}}^{-1} = 4\pi^2 \simeq 39.478$, see also \cite{AM22,EGGLS24}. Thus, our exponential growth rate is off by a factor of roughly $5.646$.

\textbf{Euler characteristic and lattice points of $\mc{M}_{g,n}$.}
Consider the Gaussian Unitary Ensemble (GUE) spectral curve $x(z) = z + z^{-1}$ and $y(z) = -z$. It has two ramification points at $z = \pm 1$. There are three interesting generalised periods \cite{Nor13,Eyn16}: 
\begin{itemize}
	\item $\chi(\mc{M}_g) = F_g$ is the Euler characteristic of the moduli space of genus $g$ curves;

	\item $\mc{N}_{g,n}(\ell_1,\ldots,\ell_n) = \big(\bigotimes_{i = 1}^n \ell_i^{-1} I_{(z,\infty,\ell_i)}\big)[\omega_{g,n}]$ is the number of metric ribbon graphs of genus $g$ with $n$ boundaries of lengths $L_1,\ldots,L_n > 0$, i.e. the number of lattice points in the combinatorial model of the moduli space of curves;

	\item $ \mc{C}_{g,n}(\ell_1,\ldots,\ell_n) = (-1)^n \big(\bigotimes_{i = 1}^n I_{(1/x,\infty,\ell_i)}\big)[\omega_{g,n}]$ is the number of maps of genus $g$ with $n$ boundaries of lengths $\ell_1,\ldots,\ell_n > 0$ and no internal face. 
\end{itemize}
The $R$-matrix in this case was computed in \cite{ACNP15}. The constants in the boundedness condition can be taken to be $M_{\theta} = 2$, $M_{\phi} = \frac{1}{4}$, and $\rho = \rho_t = 4$; the constant in the boundedness property for $I_{(z,\infty,\ell)}$ and $I_{(1/x,\infty,\ell)}$ can be taken to be $v = 4$. Thus: 
\begin{equation} 
	|\chi_{g}|
	\le
	\ms{S} \, \frac{\Gamma(2g-2)}{\ms{A}^{2g-2}}
	\qquad\text{and}\qquad
	\max\big\{\mc{N}_{g,n}(\ell_1,\ldots,\ell_n), \mc{C}_{g,n}(\ell_1,\ldots,\ell_n)\big\}
	\le
	\ms{S}_\ell \, \frac{\Gamma(2g-2+n)}{\ms{A}^{2g-2+n}}
	\,,
\end{equation}
for some $\ms{S}$, $\ms{S}_{\ell} > 0$ and $\ms{A}^{-1} = \frac{27}{10} \cdot 3^{1/3} \simeq 12.314$. It is worth mentioning that the Euler characteristic of the moduli space of curves is explicitly computed by the celebrated Harer--Zagier formula \cite{HZ86}:\
\begin{equation}
	\chi_g = \frac{B_{2g}}{2g(2g-2)} \,,
	\qquad\implies\qquad
	|\chi_g| \sim \frac{g}{\pi^2} \frac{\Gamma(2g-2)}{(2\pi)^{2g-2}} \,.
\end{equation}
Note the optimal exponential growth rate $\ms{A}_{\textup{HZ}}^{-1} = \frac{1}{2\pi} \simeq 0.159$. Thus, our exponential growth rate is off by a factor of roughly $77.372$.

\textbf{Masur--Veech volumes.}
Consider the local spectral curve given by $x(z) = z^2$, $y(z) = -z/2$, and $\omega_{0,2}(z_1,z_2) = \frac{1}{2} \bigl( \frac{1}{(z_1 - z_2)^2} + \frac{\pi^2}{\sin^2(\pi(z_1 - z_2))} \bigr) \dd z_1 \dd z_2$. It has a single ramification point at $z = 0$, so we omit the dependence on $\alpha \in \mf{a}$. The amplitudes $F_{g;0,\ldots,0}$ compute the Masur--Veech volumes of the principal stratum of the moduli space of quadratic differentials \cite{And+23}:
\begin{equation}
	\mathrm{Vol}(\mc{Q}_{g,n})
	=
	2^{4g-2+n} \frac{(4g-4+n)!}{(6g-7+2n)!} \, F_{g;0,\ldots,0} \,.
\end{equation}
The constants in the boundedness condition can be taken to be $M_{\theta} = \frac{1}{16 \sqrt{2}}$, $M_{\phi} = \frac{\pi^2}{12}$, and $\rho = \frac{1}{8}$. Thus:
\begin{equation}
	\mathrm{Vol}(\mc{Q}_{g,n})
	\le
	\frac{\ms{S}}{\ms{A}^{2g-2+n}}
\end{equation}
for $\ms{A} = \tfrac{64(24 + \pi^2)}{5} \cdot 3^{-2/3} \cdot 5^{-1/2} \simeq 93.208$. We can compare the bound with the large genus asymptotic formula proved by Aggarwal in \cite{Agg21}:
\begin{equation}
	\mathrm{Vol}(\mc{Q}_{g,n}) \sim \frac{\frac{4}{\pi} (\frac{3}{4})^{n}}{(\sqrt{6})^{2g-2+n}} \,.
\end{equation}
In other words, our exponential factor is off by a factor of more roughly $228.313$.

\textbf{Maps.}
Consider the spectral curve $x(z) = \alpha + \gamma(z + z^{-1})$ and $y(z) = \sum_{k = 1}^{d} u_k z^{-k}$. It is a deformation of the GUE spectral curve, which appears in the enumeration of maps with internal faces with Boltzmann weights. The parameters $u_k$ are algebraic functions of these Boltzmann weights \cite{Eyn16}. The weighted number of maps genus $g$ with $n$ boundaries of length $\ell_1,\ldots,\ell_n > 0$ is obtained as
\begin{equation}
	\mathcal{T}_{g,n}(\ell_1,\ldots,\ell_n) = (-1)^n\bigotimes_{i = 1}^{n} I_{(1/x,\infty,\ell_i)}[\omega_{g,n}] \,.
\end{equation}
Assume that the parameters are such that $\frac{y(z) - y(z^{-1})}{z - z^{-1}}$ has no zeroes when $|x(z)| \leq \alpha + 2\gamma$ (this implies in particular off-criticality). Then, recycling the boundedness estimates for the GUE spectral curve, we can take for constants in the boundedness property $M_{\theta} = 2|\gamma|^4 \mathfrak{m}$ where $\mathfrak{m} = \max_{|z| = 1} \big|\frac{z - z^{-1}}{y(z) - y(z^{-1})}\big|$, and $M_{\phi} = \frac{1}{4|\gamma|^{2}}$, $\rho = \rho_t = 4|\gamma|^2$. For the boundedness property of the generalised period $I_{(1/x,\infty,\ell)}$ we can take $\nu = 4|\gamma|^{-2}$. Thus,
\begin{equation}
	|\mathcal{T}_{g,n}(\ell_1,\ldots,\ell_n)| \leq \ms{S}_{n,\ell} \frac{\Gamma(2g - 2 + n)}{\ms{A}^{2g - 2 + n}}
\end{equation}
for some $\ms{S}_{\ell} > 0$ and $\ms{A}^{-1}$ specified from the above constants via the formula below \eqref{eq:bound:periods}.

%
%
%
%
%
\textbf{Stationary Gromov--Witten invariants of $\P^1$.}
Consider the spectral curve $x(z) = z + z^{-1}$ and $y(z) = \log(z)$. It has two ramification points at $z = \pm 1$. The generalised periods $I_{(1/x,\infty,k+1)}$ compute the stationary Gromov--Witten invariants of the Riemann sphere \cite{NS14,DOSS14}:
\begin{equation}
	\Braket{ \tau_{k_1}(\omega) \cdots \tau_{k_n}(\omega) }_{g,d}^{\P^1}
	=
	\left(
		\bigotimes_{i=1}^n \frac{I_{(1/x,\infty,k_i+1)}}{(k_i+1)!}
	\right)
	[\omega_{g,n}] \,.
\end{equation}
The degree $d$ is determined by $|k| = 2g-2+2d$. As in the GUE case, the constants in the boundedness condition can be taken to be $M_{\theta} = 2$, $M_{\phi} = \frac{1}{4}$, and $\rho = \rho_t = 4$; the constant in the boundedness property for $I_{(1/x,\infty,k+1)}$ can be taken to be $v = 4$. Thus: 
\begin{equation} 
	\Braket{ \tau_{k_1}(\omega) \cdots \tau_{k_n}(\omega) }_{g,d}^{\P^1}
	\le
	\ms{S}_k \, \frac{\Gamma(2g-2+n)}{\ms{A}^{2g-2+n}}
	\,,
\end{equation}
for some $\ms{S}_{k} > 0$ and $\ms{A}^{-1} = \frac{27}{10} \cdot 3^{1/3} \simeq 12.314$.

\printbibliography

\end{document}